\documentclass[pra,aps,letterpaper,10pt,twocolumn,nofootinbib,notitlepage,floatfix]{revtex4-1}

\usepackage{amsmath,amssymb,amsthm,xcolor,bm,bbm,tikz,pgfplots,algorithm,algpseudocode,graphicx,hyperref,cleveref}
\DeclareGraphicsExtensions{.pdf}
\graphicspath{ {figures/} }

\AtBeginDocument{
\heavyrulewidth=.08em
\lightrulewidth=.05em
\cmidrulewidth=.03em
\belowrulesep=.65ex
\belowbottomsep=0pt
\aboverulesep=.4ex
\abovetopsep=0pt
\cmidrulesep=\doublerulesep
\cmidrulekern=.5em
\defaultaddspace=.5em
\tabcolsep=7pt
}
\usepackage{booktabs}

\makeatletter
\newcommand\fs@booktabsruled{%
  \def\@fs@cfont{\bfseries\strut}\let\@fs@capt\floatc@ruled
  \def\@fs@pre{\hrule height\heavyrulewidth depth0pt \kern\belowrulesep}%
  \def\@fs@mid{\kern\aboverulesep\hrule height\lightrulewidth\kern\belowrulesep}%
  \def\@fs@post{\kern\aboverulesep\hrule height\heavyrulewidth\relax}%
  \let\@fs@iftopcapt\iftrue
}
\makeatother
\floatstyle{booktabsruled}
\restylefloat{algorithm}

\usepackage{etoolbox}
\AtBeginEnvironment{algorithm}{\linespread{1.25}\selectfont}

\newtheorem{theorem}{Theorem}
\newtheorem{lemma}{Lemma}
\newtheorem{corollary}{Corollary}

\definecolor{darkblue}{RGB}{0,0,127} 
\definecolor{darkgreen}{RGB}{0,150,0}
\hypersetup{breaklinks, colorlinks, linkcolor=darkblue, citecolor=darkgreen, filecolor=red, urlcolor=blue}
\hypersetup{pdftitle={Scalable Bayesian Hamiltonian Learning}, pdfauthor={Tim J. Evans, Robin Harper and Steve T. Flammia}}
\newcommand\numberthis{\addtocounter{equation}{1}\tag{\theequation}}

\newcommand{\expect}[1]{\mathbb{E}\left[#1\right]}
\newcommand{\isnorm}{\sim \mathcal{N}}
\newcommand{\T}{^\text{T}}
\newcommand{\rhoav}{\bar{\rho}}
\newcommand{\ket}[1]{\left|#1\right>}

\newcommand{\argmin}[1]{\underset{#1}{\mathrm{argmin}}}
\DeclareMathOperator{\Tr}{Tr}

\newcommand{\methodname}{Bayesian Hamiltonian Learning}
\newcommand{\method}{BHL}

\begin{document}
\title{Scalable Bayesian Hamiltonian Learning}

\author{Tim J.\ Evans}
\author{Robin Harper}
\author{Steven T.\ Flammia}
\affiliation{Centre for Engineered Quantum Systems, School of Physics, The University of Sydney, Sydney, Australia}
\date{\today}

\begin{abstract}
As the size of quantum devices continues to grow, the development of scalable methods to characterise and diagnose noise is becoming an increasingly important problem.
Recent methods have shown how to efficiently estimate Hamiltonians in principle, but they are poorly conditioned and can only characterize the system up to a scalar factor, making them difficult to use in practice.
In this work we present a Bayesian methodology, called \methodname{}~(\method{}), that addresses both of these issues by making use of any or all, of the following: well-characterised experimental control of Hamiltonian couplings, the preparation of multiple states, and the availability of any prior information for the Hamiltonian. 
Importantly, \method{} can be used online as an \emph{adaptive} measurement protocol, updating estimates and their corresponding uncertainties as experimental data become available.
In addition, we show that multiple input states and control fields enable \method{} to reconstruct Hamiltonians that are neither generic nor spatially local.
We demonstrate the scalability and accuracy of our method with numerical simulations on up to 100 qubits. 
These practical results are complemented by several theoretical contributions. 
We prove that a $k$-body Hamiltonian $H$ whose correlation matrix has a spectral gap $\Delta$ can be estimated to precision $\varepsilon$ with only $\tilde{O}\bigl(n^{3k}/(\varepsilon \Delta)^{3/2}\bigr)$ measurements.
We use two subroutines that may be of independent interest: 
First, an algorithm to approximate a steady state of $H$ starting from an arbitrary input that converges factorially in the number of samples; and second, an algorithm to estimate the expectation values of $m$ Pauli operators with weight $\le k$ to precision $\epsilon$ using only $O(\epsilon^{-2} 3^k \log m)$ measurements, which quadratically improves a recent result by Cotler and Wilczek. 
\end{abstract}

\maketitle

\section{Introduction}\label{section:introduction}

Extracting diagnostic information about noise processes is central to the development and improvement of quantum devices. 
Already we are witnessing the realisation of quantum devices that are of a size that is out of reach for standard tools for experimental noise characterisation~\cite{Preskill2018quantumcomputingin}. 
Randomized benchmarking and its variants~\cite{emerson2005scalable,Emerson2007,knill2008randomized,dankert2009exact,magesan2012characterizing} offer efficient characterisation of quantum devices through averaging the noise and consequently reducing the number of parameters to be learned. 
However, for the most part they offer performance metrics and do not pinpoint the physical origin of noise sources, giving limited diagnostic insight.

New approaches have been proposed and experimentally demonstrated that yield more detailed error models than standard randomized benchmarking, for instance allowing full reconstruction of Pauli error channels and consequently all correlated errors~\cite{erhard2019characterizing,harper2019efficient,flammia2019efficient}.
Such methods, however, remove coherent noise terms in order to maintain scalability. 
Other scalable methods have been proposed to determine all $k$-qubit reconstructed density matrices of multi-qubit systems \cite{cotler2019,garciapereze2019}, although such methods are aimed at reconstructing specific states and do not give insight as to the dynamics of the systems in question.

Hamiltonian learning is a well-studied problem~\cite{granade2012robustonline,granade2014quantumhamiltonian,wiebe2014hamiltonianlearning,wiebe2015quantumbootstrapping,krastanov2019stochastic}  that addresses the need to characterize coherent noise sources. 
In general, it requires the estimation of a number of parameters that scales exponentially in the system size, but most physically relevant Hamiltonians will have only few-body interactions and are described by polynomially many parameters. 
Reconstructing the Hamiltonian of a quantum system can provide rich diagnostic information for an experimentalist seeking to reduce noise-induced errors in a device. 
The benefits from learning the Hamiltonian extend past diagnosis, opening up a range of engineering tools that can be used to counteract, say, noise couplings. 
For instance, pulse shaping techniques such as GRAPE~\cite{khaneja2005} can be used to design specific pulse shapes, leading to vastly improved fidelities \cite{Yang2019}. 
The GRAPE optimization procedure, however, requires a good characterisation of the Hamiltonian affecting the system making it prohibitive for larger devices.

Recent works~\cite{qi2019determininglocal,bairey2019learning,garrison2018does,chertkov2018computational,hou2019determining} have shown how one can reconstruct a generic spatially local Hamiltonian given a single state that commutes with the Hamiltonian. 
These results can also be generalised to local Lindbladians~\cite{bairey2019lindblad}. 
Despite being remarkable technical results, there remain several fundamental barriers to the practical implementation of these ideas.
For example, the unknown Hamiltonian is only recovered up to a scalar factor, $ H\sim\alpha H $, and the inverse problem is generally ill-conditioned, making it highly sensitive to noise.
Also, there are many physically interesting Hamiltonians that are neither generic nor local. 

In the next section, we will review the method of Hamiltonian estimation developed in~\cite{qi2019determininglocal,bairey2019learning} and describe in more detail the barriers to making this method practical. 
We present our main results for addressing these difficulties in \cref{section:results} and provide a discussion of future directions in \cref{section:discussion}. 
In the appendices, we provide a derivation of our Bayesian model, provide proofs for our claims about the subroutines together with pseudocode for them, and prove our upper bound on the sample complexity for accurate reconstruction.

\subsection{Problem statement}\label{section:problem-setting}

Consider a $d$-dimensional quantum system consisting of $n$ finite-dimensional spins, and let $H$ be the system Hamiltonian. 
We will focus this discussion on qubit Hamiltonians, so $d=2^n$. 
Let us choose as an operator basis the $n$-qubit Pauli matrices, $\{P_j\}$ so that $P_j^{\vphantom{\dagger}}=P_j^\dagger$ and $\Tr(P_j P_k)=d\delta_{jk}$.
We can then expand 
\begin{align}
    H = \sum_{i=1}^{m} c_i P_i
\end{align}
with a vector of couplings $c \in \mathbb{R}^m$ and $c_i = \frac{1}{d}\Tr(P_i H)$. 
Any Hamiltonian can be written in such a way when $m=d^2$, but in practice most Hamiltonians have only few-body couplings and hence are well described by an expansion in a local basis having $m=O\bigl(\mathrm{poly}(n)\bigr)$. 
The most physically relevant example is the set of $n$-qubit Pauli operators that act non-trivially on only $k$ or fewer sites (a \emph{$k$-body} operator) or, more restrictively, on $k$ spatially contiguous sites (a \emph{$k$-local} operator). 
When $k=n$ we recover a general operator, but for $k=O(1)$ we can still accurately and efficiently describe generic few-body couplings, including nearly all cases of experimental relevance.

We now define a \emph{steady state} of $H$ to be any state $\rho$ such that $[H,\rho]=0$.
In general, a steady state will depend implicitly on the coupling constants $c$ that define $H$, since $[H(c_1),\rho]=0$ says nothing about the value of $[H(c_2),\rho]$.

When $\rho = \rho(c)$ is a steady state of a $k$-body Hamiltonian $H=H(c)$, it will satisfy a set of $(2k-1)$-body constraints given as follows~\cite{bairey2019learning}. 
Define the $m\times m$ matrix $K = K(\rho)$ given by
\begin{align}
\label{eqn:sqrt-correlation-matrix}
	K_{jk} : = \Tr\bigl(i[P_j,P_k]\rho\bigr)\,.
\end{align}
Then~\cite{qi2019determininglocal,bairey2019learning} the vector $c$ of couplings in $H$ is in the kernel of $K$,
\begin{align}
\label{eqn:K-kernel}
    Kc = 0\,.
\end{align}
Since the matrix elements of $K$ are all observable (they are expectations of hermitian operators), they can be estimated through a series of experiments. 
Finding an approximate null vector from a noisy $\tilde{K} \approx K$ would be an estimator $\hat{c}$ for something proportional to the unknown vector of couplings $c$. 

There are several challenges in implementing this idea in practice. 
First, preparing an appropriate steady state is challenging, since there seems to be a tradeoff in the complexity of preparing steady states and their utility in this estimation scheme: 
States that contain lots of information about $H$ (e.g.\ a ground state of $H$) seem to be difficult to prepare in general, and states that are easy to prepare in general (e.g.\ the maximally mixed state) are not useful because they are compatible with more than one $H$.

Even assuming that a suitable initial state can be prepared \emph{efficiently}, it must still be prepared \emph{accurately}. 
More generally, one must worry about state preparation and measurement errors (SPAM) and how they will affect the method.
Even in the absence of statistical noise from measurements, preparing a state that is a steady state of the wrong Hamiltonian will bias the estimate of~$c$. 
A reliable reconstruction method should be robust to small SPAM errors and should accurately quantify the error uncertainties due to the SPAM.

Next, given an estimate $\hat{K}$ of the matrix $K$, even in the absence of state preparation errors the noise will in general depend on the signal $c$. 
This means that accurately quantifying the uncertainty of an estimate becomes impossible with naive estimators. 

Furthermore, the inverse problem of recovering $\hat{c}$ from $\hat{K}$ is generally ill-posed. 
A simple estimate of how ill-posed it is comes from the spectral gap of $K$, defined as follows.  
Introduce a matrix $M = K\T K$, which is equal to the following $m\times m$ matrix that is called the \emph{correlation matrix}~\cite{qi2019determininglocal,bairey2019learning} 
\begin{align}
\label{eqn:correlation-matrix}
M_{j,k} := \frac{1}{2} \Tr\bigl( \{P_j, P_k\} \rho \bigr) - \Tr\bigl(P_j\rho \bigr) \Tr\bigl(P_k\rho \bigr)\,.
\end{align}
If the eigenvalues of $M$ are $\lambda_m \ge \ldots \ge \lambda_2 \ge \lambda_1$, then we define the \emph{spectral gap} of $K$ to be the quantity $\Delta := \Delta(K) = \lambda_2-\lambda_1$. 
This is of course also the actual spectral gap of the correlation matrix $M$, but it will be more convenient to work with $K$ throughout, so we will abuse language and call this the spectral gap of $K$.

Perhaps surprisingly, in the case of a generic $k$-body Hamiltonian the spectral gap is nonzero~\cite{qi2019determininglocal}, so the kernel of $M$ (and hence $K$) is unique and exact reconstruction of $c$ up to a scale factor is possible in the noiseless case. 
However, the gap is in many cases so small as to preclude any realistic implementation once noise is introduced. 
For example, in Ref.~\cite{bairey2019learning} unique reconstruction was achieved when the noise standard deviation was $10^{-12}$, which corresponds to about $\sim10^{24}$ or more measurements. 
While the small gap ``only'' introduces a constant overhead factor, it is unfortunately too large to allow practical reconstruction. 
Moreover, there is no compelling theoretical reason known for why the overhead introduced by the gap shouldn't also depend on the size $n$ of the Hamiltonian and possibly even exponentially. 

Finally, even after finding a good enough estimate $\hat{K}$ such that there is a unique null space spanned by $\hat{c}$, any real number $\alpha$ gives an equally valid estimate $\alpha \hat{c}$ if \cref{eqn:K-kernel} is the only reconstruction requirement. 
To eliminate this additional ambiguity, either further constraints beyond \cref{eqn:K-kernel} would be required, or a further phase estimation step would have to be performed.

\section{Results}\label{section:results}

\subsection{Summary of main results}\label{section:summaryofresults}

In this paper, we address many of the above difficulties to arrive at a protocol for learning Hamiltonians that is practical for present day experiments in quantum computing and quantum many-body physics.

Our most important contribution is to make the inverse problem well-posed. 
We achieve this by introducing two new degrees of freedom to the experimental design: the ability to choose multiple state preparations and/or multiple control fields. 
At first glance this might seem to increase, not decrease, the complexity of the inverse problem. 
In fact, the addition of even one additional state preparation or control field is enough to improve the spectral gap by many orders of magnitude in relevant situations. 
This makes the total number of measurements required to obtain a useful estimate well within the realm of feasibility for many current experiments. 
Moreover, when adding control fields, the controls themselves are also estimated by the algorithm. 
This addresses one of the central difficulties in using pulse shaping and dynamical decoupling by giving a method to efficiently obtain a complete description of all dynamical variables for a universal set of controls. 

We next show how the estimation can be cast in a Bayesian formulation to yield an algorithm we call \method{}. 
This confers a long list of advantages. 
First, by utilizing prior information the Bayesian estimation achieves a much greater speed of convergence to an improved estimate, and second it intrinsically comes with rigorously justifiable error bars directly from the data without resorting to numerically expensive and heuristic post-processing. 
Remarkably, the Bayesian framework also lets us return a \emph{point estimate} of the true couplings, removing the overall scalar factor ambiguity that has plagued previous methods~\cite{qi2019determininglocal,bairey2019learning}. 
It also allows us to correctly deal with the fact that the noise on the estimate depends on the unknown Hamiltonian itself. 
This in turn helps to avoid over-fitting and problematic bias in estimates.

We prove that the \method{} estimator is robust to a large class of measurement errors, and we address the issue of state preparation errors in two ways. 
First, we incorporate them into the Bayesian model to allow for accurate uncertainty quantification, and second, we prove that approximate state preparations can be used in conjunction with efficient time averaging~\cite{bairey2019learning} to yield improved accuracy in the estimates. 
We can perform the relevant data collection in a highly parallel fashion from single-qubit measurements only, and our algorithm for low-weight Pauli expectation value estimation improves quadratically over the recent work \cite{cotler2019}. 

Finally, we discuss how to implement an online version of \method{}. 
This enables the estimation of Hamiltonians in real-time whenever the preparation of the input steady state can be done quickly.

\subsection{Using multiple input states}\label{section:multiple-preparations}

Let us first consider what happens when we prepare multiple input states, $\rho_1,\ldots,\rho_N$, and construct their corresponding matrices $K_i := K(\rho_i)$. 
Clearly \cref{eqn:K-kernel} must hold for each $K_i$, and we can stack each of these constraints into a single matrix constraint,
\begin{align*}
A := \begin{bmatrix}
K_1\\
\vdots\\
K_N\\
\end{bmatrix}
\text{\ and\ \ }x := c\,.\numberthis\label{eqn:A-multiple-states}
\end{align*}
We adopt the notation $A$ for a composite object that incorporates multiple constraints, and we label our unknown as $x$ in this more general setting. 
The analog of \cref{eqn:K-kernel} now becomes
\begin{align}\label{eqn:A-kernel}
    Ax=0\,.
\end{align}

The spectral gap of $A$ obeys the inequality,
\begin{align}\label{eqn:gap-ineq}
    \Delta(A) \ge \max_j \Delta(K_j)\,,
\end{align}
which follows from Weyl's inequality~\cite{Bhatia1997}.
Thus the gap of $A$ is at least as good as the best constituent constraint matrix $K_j$ that comprises a block of $A$. 

\begin{figure}[t!]
	\centering
	\begin{tikzpicture}
	\node at (0,0) {\includegraphics[width=\columnwidth]{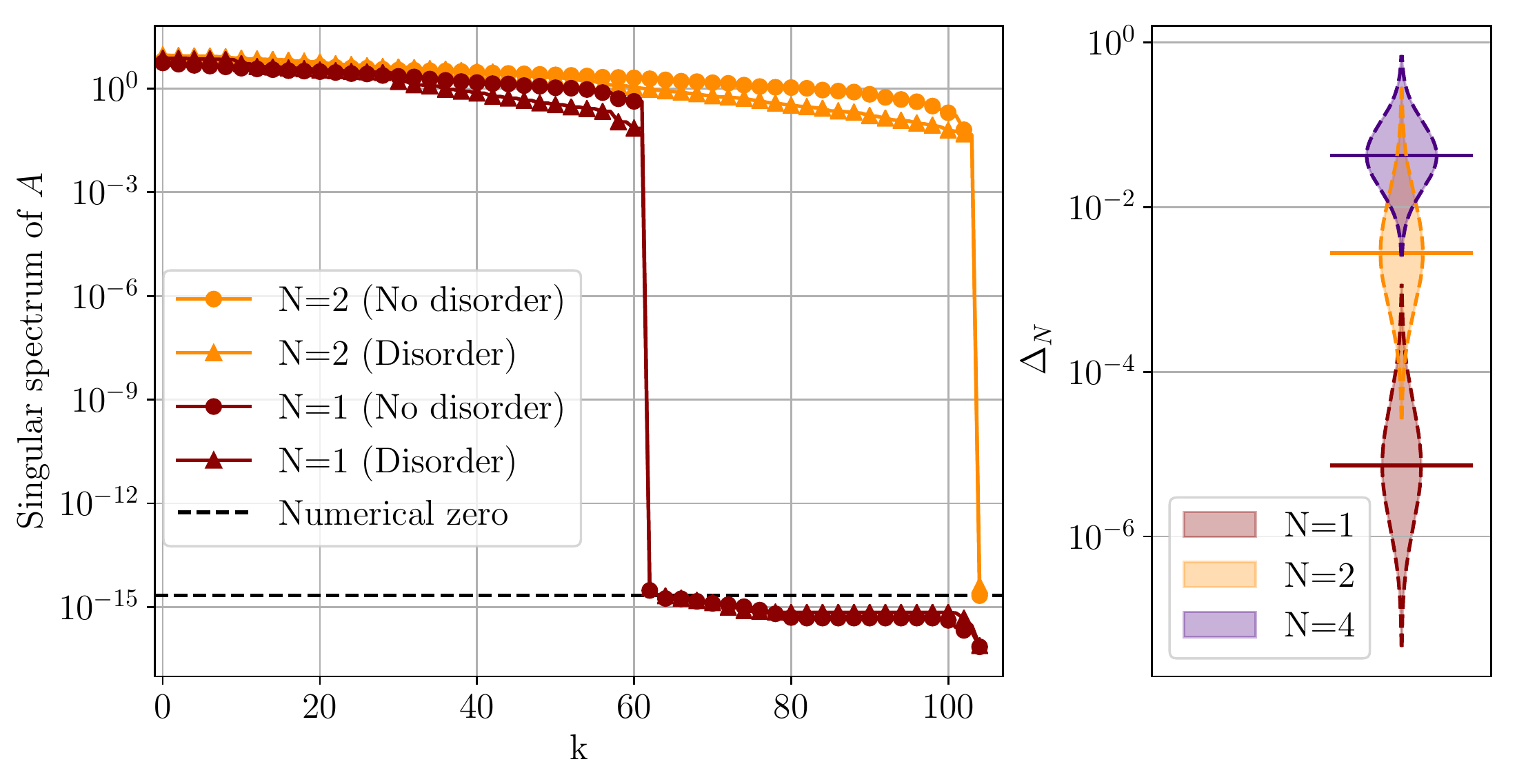}};
	\node at (-1,-2.5) {(a)};
	\node at (3.2,-2.5) {(b)};
	\end{tikzpicture}
	\caption{This figure shows how the use of multiple eigenstates can yield unique reconstruction, even of non-generic and non-local Hamiltonians. 
	In (a) we illustrate the effect of using multiple eigenstates on a randomly generated disordered and non-disordered Hamiltonian in the form of \cref{eqn:lr-ising-hamiltonian}. 
	Plotted are the singular value spectra for the operator $ A $ as more eigenstates are added to the estimation process. 
	We see that for even one additional preparation we remove the degeneracy in $ A $. 
	In (b) we plot the spectral gap $\Delta(K)$ for randomly generated 10-qubit local Hamiltonians using an increasing number of $N$ eigenstates, showing the improvement in spectral gap over a wide variety of Hamiltonians.}
	\label{fig:multiple_states_svs}
\end{figure}

In practice, the inequality (\ref{eqn:gap-ineq}) greatly understates the improvement to the spectral gap in practically relevant cases. 
We illustrate this in \cref{fig:multiple_states_svs} for two separate cases: a disordered 2-local spin chain, as was studied in ref.~\cite{qi2019determininglocal}, and a nonlocal, 2-body Ising-type model. 
The Hamiltonian for the nonlocal model is given by
\begin{align}
\label{eqn:lr-ising-hamiltonian}
    H = \sum_{i\neq j} X_i X_j + \sum_{i\neq j} J_{i,j} P_i P_j\,,
\end{align}
where the couplings $J_{i,j}$ control the strength of arbitrary two-qubit Pauli interactions. 
These couplings can be chosen to be either uniformly 0 (in the case of no disorder) or sampled as independent Gaussian random variables in the disordered case, $J_{i,j}\isnorm(0,\sigma)$; we choose $\sigma = 10^{-1}$ for the simulations in \cref{fig:multiple_states_svs}. 
The disorder is intended to avoid any special structure that might inadvertently close the spectral gap.
The Hamiltonian (\ref{eqn:lr-ising-hamiltonian}) is physically relevant because it is the coupling that drives the global entangling M\o{}lmer-S\o{}rensen gate used in ion-traps~\cite{sorenson2000entanglement}, and understanding deviations from the uniform case will help calibrate these gates. 

For our simulations of \cref{eqn:lr-ising-hamiltonian}, we prepare multiple eigenstates $\ket{E_i}$ for $i=1,...,N$, and measure each corresponding $K_i$.
Fig.~\ref{fig:multiple_states_svs} shows the singular spectrum of $A$ for $N=1,2$, or $4$ eigenstate preparations.
Firstly, we can see that for a single eigenstate there is a highly degenerate ground space, meaning there are many Hamiltonians that share this preparation as an eigenstate.
Therefore, the $N=1$ estimator will fail in this case of a nonlocal Hamiltonian, even in the presence of disorder. 
However, in \cref{fig:multiple_states_svs}(a) we see that with the addition of a single eigenstate this degeneracy is lifted and, up to numerical precision, $\dim \ker A = 1$. 
This means that the reconstruction is now unique: i.e., there exists only one $2$-body Hamiltonian that has both of those preparations as eigenstates.
\Cref{fig:multiple_states_svs}(b) shows how this holds, even for random Hamiltonians; the spectral gap improves even further with additional eigenstates.

\subsection{Using multiple control fields}\label{section:multiple_controls}

Now suppose that we also wish to utilize and characterize several additional control fields. 
We will expand these control fields in the same basis $\{P_j\}$ as before, so that for the control field $V_i$ we have
\begin{align}
    V_i = \sum_{j=1}^m v_{i,j} P_j\,,
\end{align}
where $v_i$ is the vector of coupling constants. 
In the presence of the control field the total Hamiltonian is given by
\begin{align}
    H_i := H_0 + V_i = \sum_{j=1}^m \bigl(c_j + v_{i,j}\bigr) P_j\,,
\end{align}
where $H_0$ is the bare Hamiltonian in the absence of any controls. 

Preparing a steady state of $H_i$ means that the correlation matrix $K$ has a kernel of $c+v_i$. 
Since we seek to estimate the $v_i$ as well, we can incorporate these additional variables and constraints again into a larger matrix (again called $A$) and a longer list of variables (again called $x$), given by
\begin{align}
\label{eqn:A-mult-controls}
A := \begin{bmatrix}
K_0&0&0 &\ldots&0\\
K_1&K_1&0 &\ &\ \\
K_2&0&K_2 &\ &\ \\
\vdots&\ &\ &\ddots &\vdots \\		
K_N&0&\ &\ldots&K_N \\		
\end{bmatrix}, \ \ x:=\begin{bmatrix}
c\\
v_1\\
\vdots\\
v_N\\		
\end{bmatrix}.
\end{align}
Furthermore, it is clear that using multiple control settings and multiple input states are compatible with one other, as any extra input states for each control field can again be stacked vertically onto the matrix $A$. 

Unfortunately, even if the matrices $K_j$ have a unique kernel, it is not true that the kernel of $A$ in \cref{eqn:A-mult-controls} is unique. 
If we have $N \ge 0$ separate control fields, then by a rank-counting argument there will be at least $N+1$ independent consistent solutions. 
It seems we have made an already ill-posed problem \emph{worse} by introducing the control fields! 

One possible solution is to add additional state preparations. 
In the generic case, adding $\ell$ extra state preparations will suffice to break the degeneracy, where $\ell = \lceil N/(m-1)\rceil$, which again follows from a rank-counting argument. 
However, even without these extra state preparations there is still utility in this strategy. 

To avoid the inconvenience of multiple state preparations, we can instead assume that one has access to control fields that are already well characterized. 
This is often the case in practice, such as when a well-characterized single-qubit gate is already known, but one wishes to characterize a two-qubit gate. 
Then the prior information about the well-characterized control fields serves to pin down a preferred solution within the $(N+1)$-fold degenerate space. 
To see how this works in more detail, we must first introduce our Bayesian model to properly account for this prior information.

\subsection{Bayesian model}\label{section:bayesian-model}

The models presented in the previous sections all require the inference of a vector in the kernel of an operator $ A $.
In this section we will construct a Bayesian method for this task, \method{}, giving us the ability to leverage prior information for greater robustness to noise as well as providing us with accurate quantification of the uncertainty in our estimates. 

One of the most important features of the Bayesian method is that it provides a point estimate for system couplings.
This is in contrast to prior works~\cite{bairey2019learning,bairey2019lindblad,qi2019determininglocal,hou2019determining,chertkov2018computational} that only resolve the system parameters up to a linear subspace by finding an approximate null vector.
As shown in \cref{fig:multiple_states_svs}, for physically relevant instances this subspace can even be larger than one dimensional, but even in the case of a unique kernel the prior methods failed to estimate the overall scale factor, and this had to be added by hand. 
\method{} eliminates this ambiguity. 

As with all prior related works~\cite{bairey2019learning,bairey2019lindblad,qi2019determininglocal,hou2019determining,chertkov2018computational}, our model will be $ Ax = 0 $ where our unknown $ x $ will be in the kernel of an operator $ A $.
We will call $A$ the \emph{forward operator}.

\begin{figure*}[t!]
	\centering
	\begin{tikzpicture}
	\node[inner sep=0pt] at (0,0) {\includegraphics[width=\textwidth]{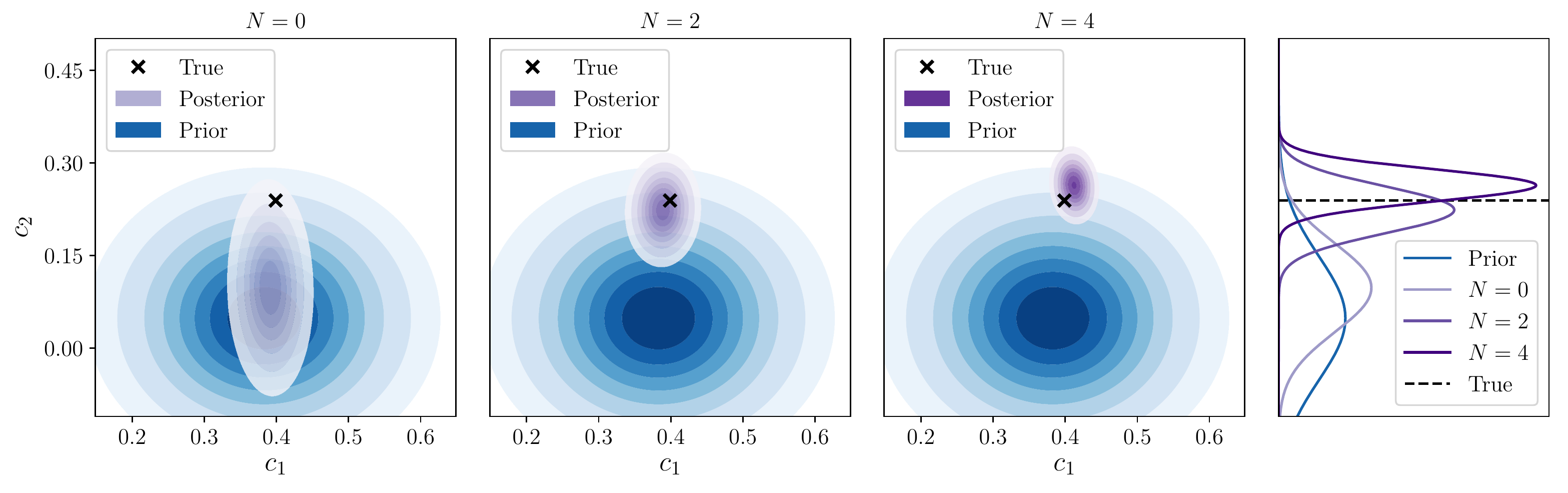}};
	\end{tikzpicture}	
	\caption{Here we plot the marginal prior and posterior distributions for the first two parameters, $ c_1,c_2 $, of a random 100 qubit spin chain Hamiltonian as more control fields are sequentially added ($ N=0,2,4 $). 
	The shaded areas represent the $2\sigma$ covariance ellipses. 
	It is clear that as more control fields are added the estimator rapidly converges towards the true values of the plotted parameters. 
	Just as importantly, the $2\sigma$ ellipse consistently contains the true values.}\label{fig:posterior}
\end{figure*}

The entries of $A$ are inherently uncertain because they can only be estimated by repeated experiments. 
To model the noise in the entries of $A$, we assume that each entry of the noisy forward operator $\tilde{A}$ is a random variable
\begin{align}
    \tilde{A}_{i,j} = A_{ij} + E_{ij}
\end{align}
where we chose the distribution to be normal with zero mean, $E_{ij}\isnorm(0,\sigma^2_{E_{ij}})$, though we note that a nonzero mean could be accommodated as well. 
A more realistic noise model would be for the $E_{ij}$ to be binomial random variables, as they are likely to be obtained by averaging experimental two-outcome measurements.
However, the Gaussian approximation will be useful theoretically for updating our prior information and will be a good approximation to the binomial case in the regime of interest. 

With this noise model in place for $\tilde{A}$, the resulting model becomes
\begin{align}
\label{eqn:approximate_model}
        \tilde{A}x + \epsilon = 0\,,
\end{align}
where we call the additive noise process $ \epsilon:=-Ex $ the \textit{approximation error}~\cite{kaipio2007statistical,kaipio2004computational,kaipio2013approximate}.
We call this an approximation error as it corresponds to an uncertainty in the operator $ A $.
Importantly, the noise on $\epsilon$ depends on the unknown $x$. 
Bayesian methods provide a natural framework for handling such errors, and correctly dealing with this state-dependent noise is a key contribution of our work.

We will assume that we have a Gaussian prior distribution for our coefficients $ x\isnorm(\bar{x},\Gamma_x) $ where $ \bar{x} $ is our best guess for the unknown coefficients and $ \Gamma_x $ is the covariance reflecting the prior uncertainty. 
In this case, as we show in Appendix \ref{appendix:bayesian}, we can model the conditional distribution $ \epsilon|x\isnorm(0,\Gamma_{\epsilon|x}) $, where
\begin{align*}
\label{eqn:approximation-error-covariance}
	\Gamma_{{\epsilon|x}_{k,l}}(x) &=\begin{cases} 
	\sigma_E^2 \left(\Tr \left[\Gamma_x\right] + \|x\|_2^2\right)\,, &k = l \\
	\sigma_E^2 \left(\Gamma_{x_{k,l}} + x_k x_l\right)\,,&k\neq l
	\end{cases}\,.\numberthis
\end{align*}
This expression for the covariance of the approximation error forms the basis of our Gaussian likelihood. 
It is easy to specialize this result to the cases of multiple input states or multiple control fields, and we do so in Appendix \ref{appendix:bayesian}. 
There are two important features of the noise in \cref{eqn:approximation-error-covariance}: it is \emph{not} independent of our unknown $x$, as noted above, and it is also correlated.

Because we have chosen a conjugate prior, we also have a Gaussian posterior $ x|A\isnorm(\mu_p,\Gamma_p) $. As shown in Appendix \ref{appendix:bayesian}, the mean and covariance of the posterior is given by
\begin{gather}
	\mu_p = \argmin{x} \|L_{\epsilon|x}(x) Ax\|_2^2 + \|L_x (x-\bar{x})\|_2^2,\label{eqn:posterior-MAP} \\
	\Gamma_p = \left(\Gamma_{x}^{-1} + A^\text{T} \Gamma_{\epsilon|x}^{-1} A \right)^{-1} \label{eqn:posterior-covariance}
\end{gather}
where $ L_{\epsilon|x},L_x $ are the Cholesky factors of $ \Gamma_{\epsilon|x}^{-1} $ and $ \Gamma_x^{-1} $ respectively, that is, $L\T L = \Gamma^{-1}$ for each respective pair $L$ and $\Gamma$. 
The estimate (\ref{eqn:posterior-MAP}) is the maximum \textit{a posteriori} (MAP) estimate and can be recognized as a generalized Tikhonov regularization.
The posterior covariance matrix \cref{eqn:posterior-covariance} gives a direct quantification of the uncertainty of the point estimate.

We can compare our estimate to the estimate obtained by taking the null space of the forward operator as in~\cite{qi2019determininglocal,bairey2019learning}.
This estimate can be written explicitly as 
\begin{align}
\label{eqn:bairey-mle}
\hat{c} = \pm\,\argmin{x} \|Ax\|_2\indent\mathrm{s.t.}\indent\|x\|=\|c\|.
\end{align}
As noted above, the normalisation constraint $\|x\|=\|c\|$ and the sign ambiguity $\pm$ both depend on the true state and will not be known in practice, so this estimate is unrealistically optimistic.
In the absence of the unrealistic norm constraint in \cref{eqn:bairey-mle}, this estimator coincides with \cref{eqn:posterior-MAP} only when we have no prior information and the approximation error is i.i.d.\ zero-mean Gaussian noise, which is never the case in practice. 
This highlights the advantage of using the \method{} estimator \cref{eqn:posterior-MAP}: by more careful computation of the statistics of the approximation error $ \epsilon $ and using prior information, the estimator \cref{eqn:posterior-MAP} avoids overfitting to the measurements.
In addition, the error quantification for our estimator is given to us directly from our posterior distribution in the form of the covariance in~\cref{eqn:posterior-covariance}, and this depends less heavily on the spectral gap than \cref{eqn:bairey-mle}.

One of the nice aspects of the Bayesian formalism is the ability to sequentially update the posterior.
This is a concept that has already been explored in the context of Hamiltonian learning~\cite{granade2012robustonline} and is likewise a good fit for \method{}.
The posterior distribution is updated having only partial data which will become the prior for the next posterior update when the data set is received.
\Cref{alg:online-learning} details the online procedure for \method{}.
\begin{algorithm}[H]
	\caption{\label{alg:online-learning} $\textsc{OBHL}(\bar{x}, \Gamma_x, A)$\\ Online \method{}}
	\begin{algorithmic}[1]
		\Require $\bar{x}, \Gamma_x, A$\\
		\# $ x\isnorm\left(\bar{x},\Gamma_x\right) $ is the initial prior distribution\\
		\# $A=\{A_1,\dots,A_N\}$ a list of $N$ measured submodels
		\For{$k=1,\dots,N$}
		\State $\bar{x}\gets \argmin{x} \|L_{\epsilon|x}(x) A_k x\|_2^2 + \|L_x (x-\bar{x})\|_2^2$
		\State $\Gamma_x\gets \left(\Gamma_{x}^{-1} + A_k^\text{T} \Gamma_{\epsilon|x}^{-1} A_k \right)^{-1}$ 
		\EndFor		
		\Ensure $(\bar{x},\Gamma_x)$
	\end{algorithmic}
\end{algorithm}

At each iteration the models $ A_k $ could be the relevant $K_i$ for different input states or an additional control setting.
Moreover, as shown in~\cite{bairey2019learning,bairey2019lindblad}, the system couplings defined on a subregion of the full system depend only on the measurements of that subregion alone.
This suggests an adaptive measurement scheme where the subsequent measurement of $ A_k $ are measurements of subregions of the full system, updating the posterior for those couplings only.

\begin{figure}[t!]
	\centering
	\begin{tikzpicture}
	\node[inner sep=0pt] at (0,0) {\includegraphics[width=\columnwidth]{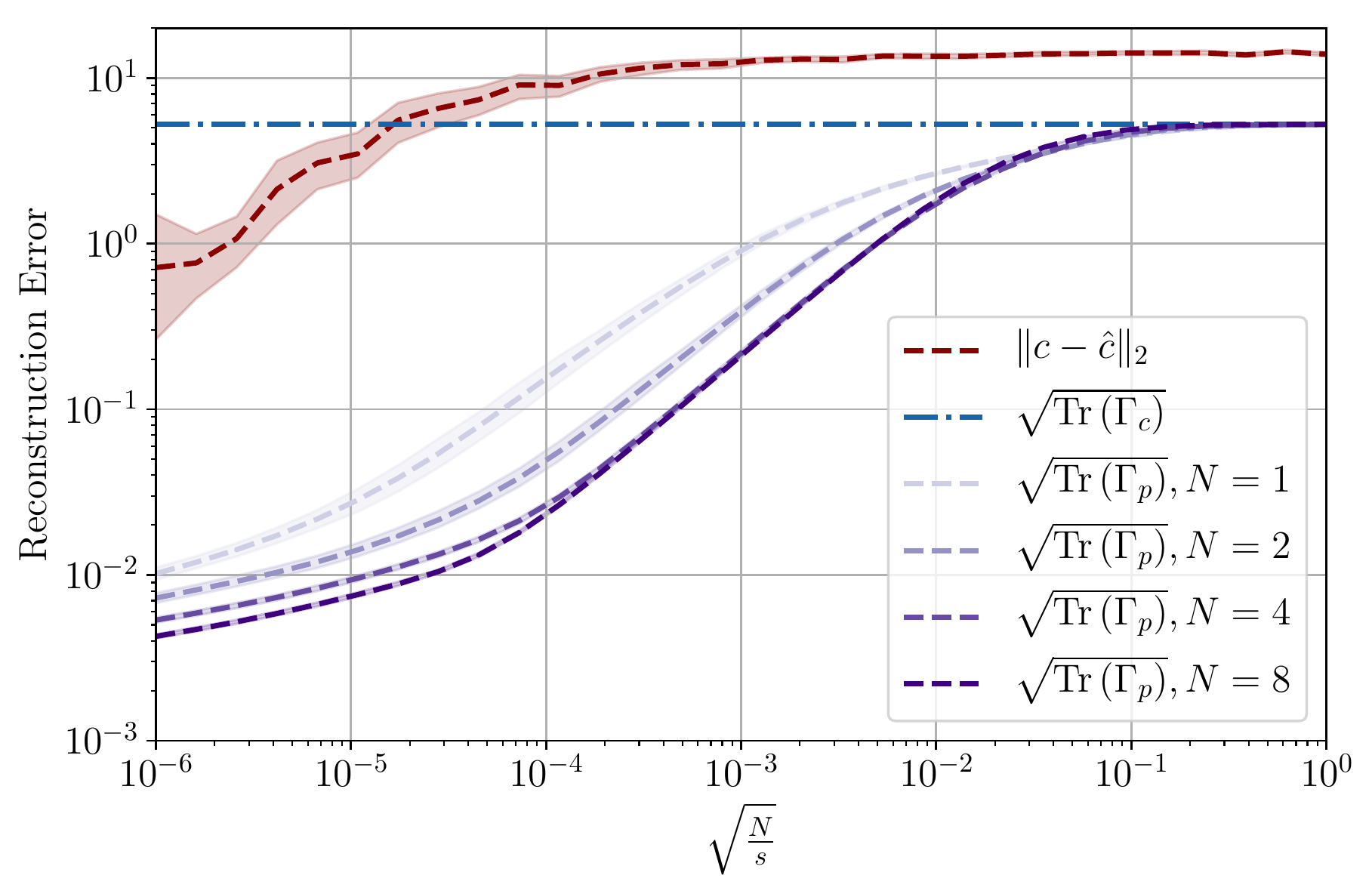}};
	\end{tikzpicture}
	\caption{Reconstruction error for different schemes as a function of measurement noise for a simulated ensemble of random 2-local Hamiltonians of 10-qubit spin chains. The red line corresponds to the estimator $ \hat{c} $ given by~\cref{eqn:bairey-mle}.
	With the \method{} plots, the Bayesian reconstruction uses a conservative prior standard deviation of $ \sigma_{c} = 0.5 $ and well calibrated control fields with $\sigma_{v_{i}} = 10^{-3}$. 
	The different error curves for additional control settings ($ N=1,2,4,8 $)  all have constant experimental cost; the experimental settings are sampled $s/N$ times.
	The Bayesian reconstruction error is always bounded above by the prior uncertainty.
	The decay in performance of \method{} for very small noise is due to an inaccuracy in the solver; theoretically the reconstruction error should saturate a $\sqrt{s/N}$ scaling according to central limit theorem rates.
    }\label{fig:error-ensemble-controls}
\end{figure}

\Cref{fig:posterior} shows the reconstruction from the MAP estimate for a 100-qubit random spin chain Hamiltonian. 
The full posterior is marginalised down to two parameters of the $1191$-dimensional full parameter space.
The couplings are known up to a prior accuracy of $\sigma_c = 10^{-1}$ and $N=4$ extra random control fields $v_i$ are sequentially added. 
These control fields are well characterised, with prior $\sigma_{v_i}=10^{-3}$. 
The Hamiltonians are represented using Matrix Product Operators (MPOs) and the corresponding reconstructions are computed using an implementation of the density matrix renormalization group (DMRG).
(For a review of MPOs and DMRG, see Ref.~\cite{bridgeman2017handwaving}.)
We can see that as the control fields are added, the posterior sequentially contracts around the true parameters. 
Moreover, the  $2$ standard deviation posterior ellipse shown consistently contains the true parameter at each iteration which is one of most appealing attributes of the Bayesian formalism.

In~\cref{fig:error-ensemble-controls} we show the reconstruction error for the case where we have multiple control fields. 
Shown are the expected reconstruction errors for an ensemble of random 10-qubit spin chain Hamiltonians. 
The expectation value for the reconstruction error $\expect{\|\hat{c} - c\|_2}$ for the posterior is given by $\sqrt{ \Tr\left[\Gamma_p\right] }$.
This reconstruction error is shown for an increasing number of $N$ control fields, with constant experimental complexity; the settings are sampled $s/N$ times so that statistical power is held fixed.
If our control fields $V_i$ are better characterised than the system we see that we get better reconstruction when our experimental resources are spread over a larger set of configurations (i.e. larger $N$).
However, there are diminishing returns for this procedure for large $N$ and the $N=8$ performs only slightly better than $N=4$.
Also shown is the expected error of the prior distribution which always gives an upper bound on the posterior accuracy.
For comparison, the estimator given by~\cref{eqn:bairey-mle} is given to show how \method{} adds robustness to noise.
Furthermore, unlike \method{}, any reconstruction using \cref{eqn:bairey-mle} needs to be normalised to match the exact 2-norm of the unknown which is not be possible in practice, hence the reconstruction error for \cref{eqn:bairey-mle} should be read as best-case. 
These results are for an additive error corresponding only to sampling statistics, however for \method{} to be of practical use we must consider additional relevant noise such as state preparation and measurement (SPAM) errors.

\subsection{Dealing with errors}\label{section:SPAM-errors}

In this section we will consider the three types of errors that plague any estimation scheme, describe our strategies for overcoming them, and quantify an upper bound on the experimental effort required to suppress them.
The errors we consider are state-preparation errors, measurements errors, and statistical errors from finite sampling. 

The strategy for dealing with measurement and statistical errors is simple: take more samples. 
However, our method will measure multiple weight-$k$ Pauli expectation values in parallel and takes time $O_k(\log m)$ to measure $m$ Paulis to fixed precision, beating the naive scaling of $O(m)$ and improving on recent work achieving $O_k\bigl[(\log m)^2\bigr]$~\cite{cotler2019}, where $O_k$ means the constant implied by the big-$O$ notation depends on $k$.

We deal with state preparation errors in two steps. Following Ref.~\cite{bairey2019learning}, we first use a time-averaged state to get our initial state closer to a steady state, and then we use quadrature rules to approximate this time average.

\subsubsection{Measurement errors}\label{section:measurement-errors}

Let us first consider measurement errors since these are the easiest to address. 
For a given measurement setting $ C_{jk} := i\left[P_j,P_k\right] $ we take $ s $ repeated measurements to estimate $ K_{jk} = \Tr\left(C_{ij} \rho \right) $.
However, these estimates will contain errors beyond those we obtain from finite averaging of measurement outcomes.
Let us consider the case that the $C_{jk}$ are qubit Pauli operators, and we assume that these are measured via two-outcome measurements.
Suppose that with some probability $e_{jk}$ we have a measurement error, meaning that we should have observed outcome $+1$ but instead we observe outcome $-1$ (or vice versa, symmetrically).
Subjected to this binary symmetric noise channel, the expectation value of repeated measurements becomes
\begin{align*}
K_{ij} = \left(1-2e_{ij}\right) \Tr\left(C_{ij} \rho\right).
\end{align*}

In the case where we have a uniform error rate $ e_{ij} = e $ for all $ i,j $ we obtain measurements $ K_\text{meas}=(1-2e)K $.
Therefore, since the kernel is invariant under scalar multiplication, our model is in fact inherently robust to such measurement errors, and only fluctuations between the $e_{ij}$ will contribute bias to the estimate.

That is not to say however that the estimator is invariant under measurement errors.
Measurement errors will cause the MAP estimate (\ref{eqn:posterior-MAP}) and posterior covariance (\ref{eqn:posterior-covariance}) to shift toward the prior as expected. 
Although there is no bias, more samples will be needed to achieve the same statistical resolution. 

The Bayesian methodology also allows us to incorporate a model of these measurement errors and account for them. 
Suppose that we have a statistical description of our measurement error rates $ e_{ij} $.
Then these measurement errors can also be corrected for by using the methods of Ref.~\cite{ferrie2012estimating}.
The corresponding uncertainty in this correction step can then replace the additive noise process $\epsilon$ (via~\cref{eqn:approximate_model}), supplying the relevant error $ E $ for the approximation error.

\begin{figure*}[tph]
	\centering
	\begin{tikzpicture}
	\node[inner sep=0pt] at (0,0) {\includegraphics[width=\textwidth]{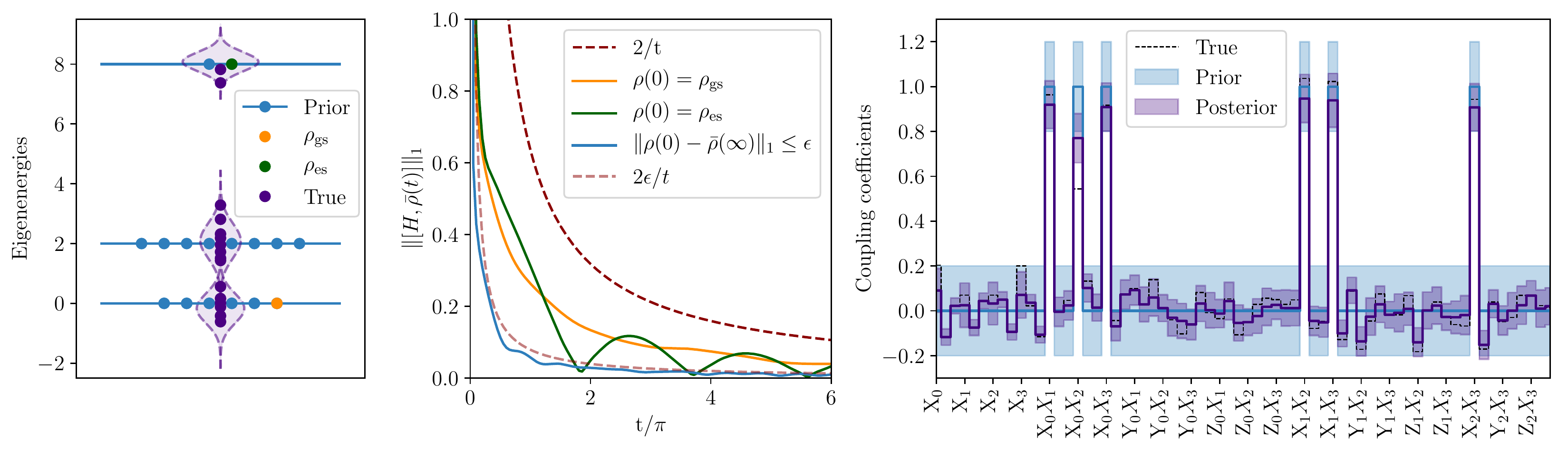}};
	\node[] at (-6.5,-3) {(a)};
	\node[] at (-1.6,-3) {(b)};	
	\node[] at (5,-3) {(c)};		
	\end{tikzpicture}
	\caption{a) The spectra for the prior mean which is centered at the Ising model in~\cref{eqn:lr-ising-hamiltonian} compared to the `true' Hamiltonian.
	The energy levels for the prior are all degenerate, and the small deviation from the ideal shows a lifting of these degeneracies. 
	The violin plots show the distribution of eigenvalues for Hamiltonians sampled from the prior.  
	b) Time-averaged input states approximately commute with the Hamiltonian. 
	By increasing the length of time we can improve the quality of the input state to the reconstruction algorithm. 
	Moreover, if we start with an initial state that is already close to commuting, such that $\bigl\|\rho(0)-\rhoav(\infty)\bigr\|_1 \le \epsilon$, then convergence is accelerated by a factor of $\epsilon$.
	c) The true $c$, prior $\bar{c}$ and posterior $\mu_p$ coupling coefficients are shown and the shaded regions are the $2\sigma$ error bars. 
	These results make use of 16 approximate time-averaged states as inputs.}\label{fig:ising-example}
\end{figure*}
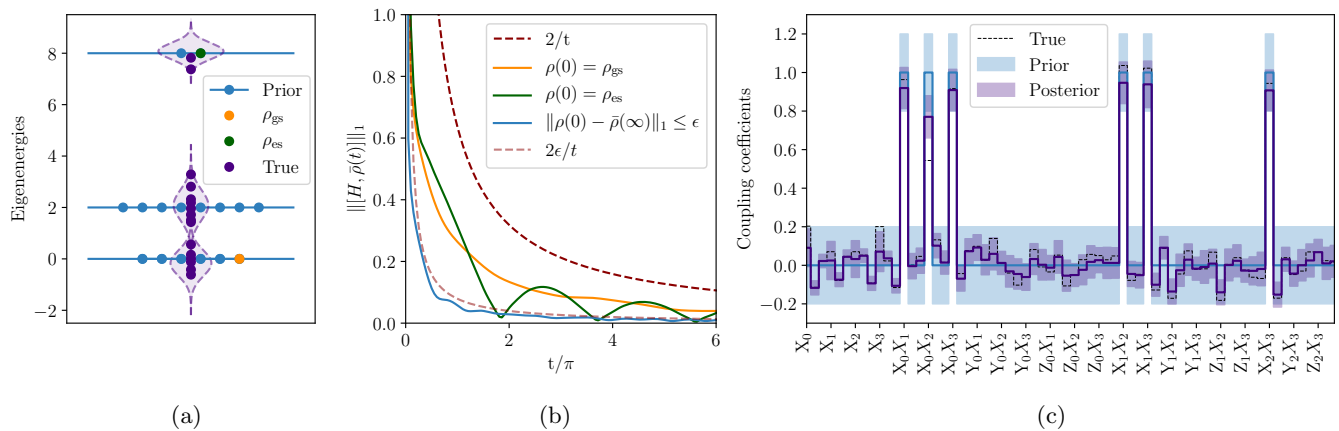

\subsubsection{State preparation errors}\label{section:preparation-errors}

Preparation errors usually pertain to our inability to exactly prepare some target state. 
In our case we won't know \textit{a priori} the state we want to prepare (that would require knowing the Hamiltonian) so preparation errors will actually correspond to the use of an approximate input state. 
This means that it is imperative to model and include preparation errors into the noise model as these are likely to be a dominant noise source in practice.
It also means that we should appeal to the system for an appropriate preparation, rather than predefine it based on prior information.

An example of a system-defined preparation is the time-averaged state $ \rhoav(t) $, introduced in~\cite{bairey2019learning}, which will approximately commute with the Hamiltonian.
It is defined by
\begin{align}\label{eqn:time-averaged-state}
\rhoav(t) :=\frac{1}{t} \int_{0}^{t} \mathrm{d}u \rho(u).
\end{align}
The appeal of this state is that the dynamics of the system provides us with the state.

Following~\cite{bairey2019learning}, we show in \cref{appendix:approximate-state-preps} that a ``warm start'' initial state that is $\epsilon$-close in 1-norm to a steady state converges like $O(\epsilon/t)$ to a true steady state. 
Therefore we can use approximate steady state preparations and time averaging to systematically reduce our state preparation errors. 

A further difficulty, is that \cref{eqn:time-averaged-state} cannot be physically realized exactly. 
We must sample from times in the interval $[0,t]$ and average an ensemble of measurements over those times to approximate the statistics from $\rhoav(t)$. 
To deal with this additional source of error, which we call quadrature error, we sample from specific times $t_i \in (0,t)$ and produce a weighted average $\rhoav_s(t) = \sum_{i=1}^s w_i \rho(t_i)$. 
In \cref{appendix:approximate-state-preps} we show that our state $\rhoav_s(t)$ rapidly converges to the desired state $\rhoav(t)$ after $s$ samples, and satisfies the bound
\begin{align}
    \bigl\|\rhoav(t) - \rhoav_s(t)\bigr\|_1 
    \le \frac{\sqrt{\pi}}{4\sqrt{s}} \biggl(\frac{\mathrm{e}\|H\|t}{4s}\biggr)^{2s}\,.
\end{align}

\subsubsection{Sample complexity upper bound}\label{sec:samplecomplexitymaintext}

We can estimate the experimental cost of Hamiltonian learning using our scheme for state preparation via time averaging and quadrature together with our method for efficient estimation of Pauli expectation values which are both detailed in pseudocode in \cref{tab:pseudocode} in \cref{appendix:algorithms}. 

We have seen numerically that the Bayesian approach outperforms the naive approximate kernel estimate \cref{eqn:bairey-mle}. 
We would like to be able to properly bound the sample complexity of this estimator. 
However, the Bayesian model is difficult to analyze because of the role played by the prior information. 
We therefore give an analysis in \cref{appendix:sample-complexity} of the sample complexity for the naive kernel estimator (\cref{eqn:bairey-mle}) to obtain an infidelity of at most $\epsilon$. 
Here the infidelity (one minus the fidelity, $F$) between two vectors $a$ and $b$ is defined as
\begin{align}
    1-F(a,b) = 1- \frac{ |a\T b|^2}{\|a\|_2^2 \|b\|_2^2}\,,
\end{align}
and is a measure that is insensitive to an overall rescaling of $a$ or $b$ by nonzero real numbers. 

In \cref{appendix:sample-complexity}, we show that starting with any \textit{a priori} bound $\max_i |c_i| \le O(1)$, we can (with high probability) learn an estimate $\hat{c}$ such that $1-F(c,\hat{c}) \le \varepsilon$ using at most
\begin{align}
    N = \tilde{O}\biggl(\frac{m^3}{\varepsilon^{3/2} \Delta^{3/2}} 3^{2k}\biggr)
\end{align}
measurements, where $\Delta$ is the spectral gap of the (noisy) forward operator, $m$ is the length of $c(H)$, and $\tilde{O}$ means we are ignoring logarithmic terms. 
A more precise version of this claim is given in \cref{thm:bound} and \cref{cor:KE}. 

The main take away from this bound is that Hamiltonian estimation of $k$-body Hamiltonians can be done in polynomial time in the number of qubits, so long as the gap $\Delta$ is not too small.
We also stress these are recovery guarantees and provide only sufficient conditions for convergence. 
Using the Bayesian methodology of \method{} allows us to leverage the use of prior information and multiple state estimates to improve the estimation dramatically beyond what can be rigorously proven.

\subsection{Example: Long-range Ising Hamiltonian}\label{section:molmer-sorenson-example}

Consider again the long-range Ising model in~\cref{eqn:lr-ising-hamiltonian} and suppose we are trying to reconstruct this from multiple time-averaged states.
We fix the spin-chain length to be $n=4$ and the basis $ \{P_i\} $ to be all $2-$body Pauli operators.
(The small value of $n$ is chosen for visual clarity in \cref{fig:ising-example}.)
We set our prior mean to be centred at the Ising model with
\begin{align*}
\bar{c}_i = \begin{cases}
1,&P_i = X_jX_k\\
0,&\text{otherwise}
\end{cases}
\end{align*}
and we choose an i.i.d.\ prior with $ \sigma_{c} = 10^{-1} $.
The true Hamiltonian $H$ is then taken to be sampled from the prior distribution.

\Cref{fig:ising-example} a) shows the distribution of the spectra for Hamiltonians drawn from the prior, centred around the degenerate mean. 
Eigenstates of our prior are natural choices for initial states as they are potentially close to steady states, especially in the presence of more prior information.
\Cref{fig:ising-example} b) shows the decay of $\|H,\rhoav(t)\|_1$ over time for different initial states including two eigenstates of the prior $\rho_\mathrm{gs}$ and $\rho_\mathrm{es}$ defined by
\begin{align*}
	\rho_\mathrm{gs} =& \ket{0001} + \ket{0100} + \ket{1011} + \ket{1110}\,,\\
	\rho_\mathrm{es} =& \ket{0000} + \ket{0011} + \ket{0101} + \ket{0110} + \\
	&\ \ket{1001} + \ket{1010} + \ket{1100} + \ket{1111} \,,
\end{align*}
which are shown in yellow and green, respectively.

The time-average of both these states approaches a commuting state according to a $1/t$ scaling, however the time-averaged state corresponding to the initial state $ \rho_\mathrm{es}$, which we will denote $\rhoav_\mathrm{es}(t) $, vanishes periodically.
If we look at the spectrum in~\cref{fig:ising-example} a), we can see that $\rho_\mathrm{es}$ has large overlap with the two most excited states.
The frequency components $f_{ij}$ of the time evolution correspond to the energy spacings of the Hamiltonian, $ \vert E_i - E_j \vert = f_{ij} \hbar $ and high frequency components decay under averaging fastest, $\sim 1/(f_{ij} t)$.
Hence, the major frequency component that persists in $ \rhoav_\mathrm{es}(t)$ is due to the difference between the two largest eigenvalues $\vert E_i - E_j \vert\approx 0.5$ which gives a period of $T/\pi\approx 2$. 
Therefore, certain initial states will yield time averaged states that approximately commute at specific times according to $\rho(0) \approx \rho(t)$.

In this example we prepare the $16$ eigenstates of the prior distribution, and time-average them for $t=3\pi$.
In between each state, we update the posterior online according to~\cref{alg:online-learning}.
Most importantly, the approximation error for the use of time-averaged states is accounted for as detailed further in \cref{appendix:approximate-state-preps}.
\Cref{fig:ising-example} c) shows the true, prior and posterior couplings for the Hamiltonian.
We see that correctly handling the approximation error allows for the use of non-commuting states and \method{} will still yield robust error quantification.
The shaded areas show the marginal variances which correspond to the diagonal of the covariance $\Gamma_p$.
However, it should be noted that \method{} provides access to the full covariance matrix, including off-diagonal correlations in the unknowns.
There are 3 couplings, the $X_0$, $X_3$, and $X_0X_2$ terms, out of the full 66 that lie outside the posterior $2\sigma$ error bars.
These errors, however, are within what we expect from a 95\% credible region. 

\section{Discussion}\label{section:discussion}

We have introduced a new framework which allows for the efficient reconstruction of $k$-body Hamiltonians called \method{}.
Our method uses prior information of the type that is likely to be available to most experimentalists.
This prior information allows us to lift the estimator from a subspace containing the system coefficients to a bare point estimate, and includes a natural quantification of the uncertainty in the form of a posterior covariance matrix.
We propose two new extended models that can be used singly or together to improve the stability and accuracy of the model, one using multiple input states and the other making use of well calibrated control fields.
We also present an algorithm~(\cref{alg:online-learning}) for performing \method{} online and suggest one way in which measurements could be performed adaptively to maximise estimation accuracy.

The major contribution that \method{} brings to Hamiltonian learning is a rigorous Bayesian framework for handling the numerous approximation errors inherent in these correlation matrix models.
We show how these errors can be systematically added to the likelihood of our model in order to prevent overfitting and give concrete examples.
Most importantly \method{} provides robust uncertainty quantification in the inference of the system parameters.

We have furthermore introduced an efficient method for approximate preparation of time-averaged steady states and an efficient method for estimating the expectation values of $m$ low-weight Pauli operators in $O(\log m)$ time. 
These two subroutines have enabled us to give a rigorous upper bound on the required sample complexity for Hamiltonian estimation, and this complexity is polynomial in $n$ for $k$-body Hamiltonians whenever the spectral gap of $K$ is at least $1/\textrm{poly}(n)$.

There are many avenues for future work. 
One obvious step is to generalize \method{} to handle Lindbladian estimation~\cite{bairey2019lindblad}. 
Another is to consider Hamiltonians for systems other than qubits. 
In particular, it would be interesting to generalize these methods to Hamiltonians and Lindbladians that are unbounded operators such as systems comprised of coupled oscillators. 

It is unclear how tight the sample complexity bounds derived here are, or if they can be significantly improved by either a better analysis of the existing algorithms or by better estimations schemes. 
In particular, it would be interesting to directly analyze \method.
It would also be interesting to find lower bounds on the sample complexity of Hamiltonian estimation. 

Finally, perhaps the most interesting direction for future work is to demonstrate the usefulness of \method{} in a real experiment.

\begin{acknowledgements}
The authors would like to thank Arne Grimsmo, Kamil Korzekwa and Phillipp Schindler for insightful discussions.
This work was supported by the Australian Research Council via EQuS project number CE170100009 and by the US Army Research Office grant numbers W911NF-14-1-0098 and W911NF-14-1-0103.
\end{acknowledgements}

\appendix
\setcounter{equation}{0}
\setcounter{figure}{0}

\section{Bayesian Preliminaries}\label{appendix:bayesian}

\subsection{Construction of the likelihood}\label{appendix:construction-likelihood}

We will begin by constructing the likelihood for our problem.
Recall that in the ideal noiseless case we have a true operator $ A $ such that $ Ax=0  $.
However, we only have access to the measured approximation $ \tilde{A} $ where they differ by some additive noise matrix $ E = A - \tilde{A} $.
This means that the error we incur in our true model is given by
\begin{align}
    \tilde{A}x &= Ax + Ex
\end{align}
leaving us with an additive noise process $ \epsilon:=-Ex $ in our model
\begin{align}
    \tilde{A}x + \epsilon &= 0
\end{align}
which we call the \textit{approximation error}~\cite{kaipio2007statistical}.

We begin with the assumption that the distribution $ \pi\left(A,x,\epsilon\right) $ is jointly normal.
Then through repeated use of Bayes' theorem,
\begin{align}
\begin{split}
\label{eqn:repeated-bayes}
    \pi\left(A,x,\epsilon\right)&= \pi\left(A|x,\epsilon\right)\pi\left(\epsilon|x\right)\pi\left(x\right)\\
    &= \pi\left(A,\epsilon|x\right)\pi\left(x\right) 
\end{split}
\end{align}
and also using the fact that
\begin{align}
\label{eqn:dirac-distribution}
    \pi\left(A|x,\epsilon\right) &= \delta\left(-Ax-\epsilon\right) 
\end{align}
we can marginalise over our approximation error $ \epsilon $ to yield our likelihood
\begin{align}
\label{eqn:likelihood}
    \pi\left(A|x\right) &= \int \pi\left(A,\epsilon|x\right)d\epsilon\nonumber\\
    &= \int \pi\left(A|x,\epsilon\right) \pi\left(\epsilon|x\right) d\epsilon, &\bigl[\text{by~(\ref{eqn:repeated-bayes})}\bigr]\nonumber\\	
    &= \int \delta\left(-Ax-\epsilon\right) \pi\left(\epsilon|x\right) d\epsilon, &\bigl[\text{by~(\ref{eqn:dirac-distribution})}\bigr]\nonumber\\	
    &= \pi_{\epsilon|x}\left(-Ax | x\right)\,.
\end{align}
Here $ \pi_{\epsilon|x}\left(-Ax | x\right) $ is the distribution of $ \epsilon|x $ \textit{evaluated} at the point $ \epsilon = -Ax $.
In the usual construction of the likelihood for such a linear model, we assume that our additive noise process and the unknown are mutually independent, in which case conditioning on $ x $ makes no difference and we are simply left with the likelihood depending on the distribution of $ \epsilon $.
However, our approximation error is $ \epsilon=Ex $, therefore we cannot make this assumption and have to carry the extra baggage of the conditional distribution in our computations.

\subsection{Gaussian Posterior}

In this section we will derive our MAP estimator from a conjugate Gaussian prior. 
In general, the computation of a posterior cannot be tractably computed, and for small parameter spaces one must often resort to sampling methods such as Markov chain Monte Carlo (MCMC)~\cite{kaipio2004computational}. 
However, for certain combinations of likelihoods and prior distributions the posterior will be exactly computable, a situation known as a \emph{conjugate prior}. 
The conjugate prior for a Gaussian likelihood is a Gaussian, and we have shown in the preceding section that our likelihood is Gaussian. 
Now we will compute the posterior distribution under a Gaussian prior, $x\isnorm(\bar{x},\Gamma_x)$, where
\begin{align}
    \pi(x|A)\propto&\pi(A|x)\pi(x)\nonumber\\
    \propto&\exp\biggl(-\frac{1}{2}(-Ax)^\mathrm{T} \Gamma_{\epsilon|x}^{-1} (-Ax) \biggr)\nonumber\\
    & \ \times \exp\biggl(-\frac{1}{2}(x-\bar{x})^\mathrm{T}\Gamma_{x}^{-1}(x-\bar{x}) \biggr)\nonumber\\
\begin{split}
    =&\exp\Biggl(-\frac{1}{2}\biggl( (-Ax)^\mathrm{T} \Gamma_{\epsilon|x}^{-1} (-Ax) +\dots \\
    &\indent(x-\bar{x})^\mathrm{T}\Gamma_{x}^{-1}(x-\bar{x}) \biggr) \Biggr)\,.
\end{split}
\end{align}
Now if we take a Cholesky decomposition of the positive definite covariance matrices $L_{\epsilon|x}^\mathrm{T}L_{\epsilon|x} = \Gamma_{\epsilon|x}^{-1}$ and $L_{x}^\mathrm{T}L_{x} = \Gamma_{x}^{-1}$ we find
\begin{align}
\label{eqn:posterior-potential}
    \pi(x|A)\propto&\exp\left(-\frac{1}{2}\|L_{\epsilon|x}(x) Ax\|_2^2 -\frac{1}{2} \|L_x (x-\bar{x})\|_2^2 \right)\,.
\end{align}
The posterior mean is then given by
\begin{align}
    \mu_p = \argmin{x} \|L_{\epsilon|x}(x) Ax\|_2^2 + \|L_x (x-\bar{x})\|_2^2.
\end{align}
Because the posterior is Gaussian we can compare \cref{eqn:posterior-potential} with a standard multivariate distribution to solve for the covariance,
\begin{align}
    \Gamma_p = \left(\Gamma_{x}^{-1} + A^\text{T} \Gamma_{\epsilon|x}^{-1} A \right)^{-1}.
\end{align}

\subsection{Distribution of approximation error}

We need to determine the statistics of the conditional distribution $\epsilon|x $ in order to be able to evaluate our likelihood in~(\ref{eqn:likelihood}).
First we will need to compute the statistics of $ \epsilon $.
Given a simple noise model corresponding to the averaging of measurement outcomes from $ s $ samples, we can assume we have zero-mean Gaussian noise in the entries of $ A $ given by $ E_{i,j} = E_{j,i} \isnorm\left(0,\sigma_E^2\right) $ where $ \sigma_E^2\approx s^{-1}$ due to the central limit theorem. 

The mean of the approximation error, using the independence of $ x $ and $ E $, is 
\begin{align}
    \bar{\epsilon} = \mathbb{E}\left[Ex\right] = \mathbb{E}\left[E\right]\mathbb{E}\left[x\right] = 0.
\end{align}
Now consider
\begin{align*}
\Gamma_\epsilon &= \mathbb{E}\left[(\epsilon-\bar{\epsilon})(\epsilon-\bar{\epsilon})\T\right]\\
&= \mathbb{E}\left[\epsilon\epsilon\T\right]\\
&= \mathbb{E}\left[Exx\T E\T\right]
\end{align*}
Let $ e_k,e_l $ be rows of $ E $. 
Then the entries of $ \Gamma_\epsilon $ are
\begin{align*}
\Gamma_{\epsilon_{k,l}} &= \mathbb{E}\left[e_k x x\T e_l\T\right]\\
&=\sum_{i,j} \expect{e_{k,i}e_{l,j}x_i x_j }
\end{align*}
and given
\begin{align*}
\expect{e_{k,i}e_{l,j} }&=\begin{cases}
\sigma_E^2&k=l,i=j\\
\sigma_E^2&k\neq l,i=l,j=k\\
0&\text{otherwise}\,,
\end{cases}
\end{align*}
hence we have
\begin{align*}
\Gamma_{\epsilon_{k,l}}(x) &= \begin{cases}
\sigma_E^2 \expect{x\T x} &k=l\\
\sigma_E^2 \expect{x_k x_l} & k\neq l\\
\end{cases}\numberthis\label{eqn:useful-gammae}\\ 
&= \begin{cases}
\sigma_E^2 \left(\Tr\left[\Gamma_x\right] + \|\bar{x}\|_2^2\right) &k=l\\
\sigma_E^2 \left(\Gamma_{x_{k,l}} + \bar{x}_k\bar{x}_k\right) & k\neq l\,.
\end{cases}\numberthis\label{eqn:prior-gammae}
\end{align*}
Now we can use the standard Schur complement computation of the conditional Gaussian distribution~\cite{kaipio2004computational} to determine our likelihood.
Let $ \epsilon|x\isnorm(\mu_{\epsilon|x},\Gamma_{\epsilon|x}) $ then
\begin{align}
\mu_{\epsilon|x} &= \bar{\epsilon} + \Gamma_{\epsilon x} \Gamma_x^{-1} (x-\bar{x})\\
\Gamma_{\epsilon|x} &= \Gamma_\epsilon + \Gamma_{\epsilon x} \Gamma_{x}^{-1} \Gamma_{x \epsilon}\,.
\end{align}
However, we have
\begin{align}
\begin{split}
\Gamma_{\epsilon x} &= \expect{\left(\epsilon-\bar{\epsilon}\right)\left(x-\bar{x}\right)\T}\\
&= \expect{E\left(x-\bar{x}\right)\left(x-\bar{x}\right)\T}\\
&= \expect{E}\Gamma_x\\
&= 0.
\end{split}
\end{align} 
This means that we obtain the conditional distribution $ \epsilon|x\isnorm(0,\Gamma_{\epsilon|x}) $ where
\begin{align}\label{eqn:conditional-covariance}
\Gamma_{\epsilon|x} &= \Gamma_{\epsilon}
\end{align}
as defined in \cref{eqn:prior-gammae}.

We can now provide specific Bayesian formulations for the two models eqs.~(\ref{eqn:A-multiple-states}) and (\ref{eqn:A-mult-controls}). 
First let us consider the model (\ref{eqn:A-multiple-states}) for the case of multiple input states.
Our prior statistics are simply $\bar{x} = \bar{c},\ \Gamma_x = \Gamma_c$ and our noise covariance  is given by
\begin{align}
\Gamma_{\epsilon|x} := \begin{bmatrix}
	\Gamma_{\epsilon_1|c}& &\ldots&0\\
	\ &\Gamma_{\epsilon_2|c}&\ &\ \\
	\vdots&\ &\ddots &\vdots \\		
	0&\ &\ldots &\Gamma_{\epsilon_N|c}\\	
	\end{bmatrix}.    
\end{align}

Next, the model for multiple control fields will have a joint prior distribution given by
\begin{align*}
\bar{x}:=\begin{bmatrix}
\bar{c}\\
\bar{v}_1\\
\vdots\\
\bar{v}_N\\		
\end{bmatrix}\text{ and }
\Gamma_x := \begin{bmatrix}
\Gamma_c& &\ldots&0\\
\ &\Gamma_{v_1}&\ &\ \\
\vdots&\ &\ddots &\vdots \\		
0&\ &\ldots &\Gamma_{v_{N}} \\		
\end{bmatrix}.\numberthis \label{eqn:multiplecontrolmodel}
\end{align*}
The noise $\epsilon$ will be distributed as $ \epsilon\isnorm(0,\Gamma_\epsilon) $ where, similar to above, the covariance $ \Gamma_\epsilon $ is the block diagonal matrix 
\begin{align*}
\Gamma_{\epsilon|x} := \begin{bmatrix}
\Gamma_{\epsilon_0|c}& &\ldots&0\\
\ &\Gamma_{\epsilon_1|c,v_1}&\ &\ \\
\vdots&\ &\ddots &\vdots \\		
0&\ &\ldots &\Gamma_{\epsilon_N|c,v_{N}} \\		
\end{bmatrix}.
\end{align*}
The conditional covariance is given by
\begin{align}
\Gamma_{{\epsilon_i|c,v_i}_{k,l}} &=\begin{cases} 
\sigma_E^2 \bigl(\Tr\left[\Gamma_c \right] + \|c\|_2^2+\\
\ \ldots\Tr\left[\Gamma_v \right] + \|v_i\|_2^2\bigr)\,, & k = l \\
\sigma_E^2 \bigl(\Gamma_{c_{k,l}} + c_k c_l\\
\ \ldots \Gamma_{v_{i_{k,l}}} + v_{i_{k}} v_{i_{l}}\bigr)\,, & k\neq l.
\end{cases}
\end{align}

\section{Approximate state preparations}\label{appendix:approximate-state-preps}

Here we show that the vector $c(H)$ of the couplings of the unknown Hamiltonian $H$ is still approximately in the kernel of a matrix $K$ when $K = K(\rho)$ is assumed only to come from a \emph{$\delta$-approximate steady state}.
By this we mean a state having a small commutator with $H$, $\|[\rho,H]\|_1\le \delta$.
Such states can be prepared using the time averaging argument first described in Ref.~\cite{bairey2019learning}. 
Our first result is a slight refinement of a similar result in~\cite{bairey2019learning}, and is stated precisely in \cref{lemma:preparation-error}.
In \cref{thm:bound} below, we will derive a bound on the error for the kernel estimator due to these approximate state preparations. 

\begin{lemma}
\label{lemma:preparation-error}
Let $\rho$ be a steady state of $H$ and suppose $\|\rho(0)-\rho\|_1 \le \epsilon$ for some $\rho(0)$. 
Then the time-averaged state $\rhoav(t)$ satisfies the inequality
\begin{align}
    \bigl\|[H,\rhoav(t)] \bigr\|_1 \le \frac{2\epsilon}{t}\,.
\end{align}
\end{lemma}

\begin{proof}
We define the time-averaged state under the evolution of $H$,
\begin{align}
\label{eq:timeavgrho}
    \rhoav(t) := \frac{1}{t} \int_0^t \mathrm{d}u \, \rho(t) = \frac{1}{t} \int_0^t \mathrm{d}u \, \mathrm{e}^{-iHu} \rho(0) \mathrm{e}^{iHu}\,.
\end{align}
This time-averaged state is also always close to a steady state, as follows. 
Using the triangle inequality and the unitary invariance of the norm, we have
\begin{align}\begin{split}
        \|\rhoav(t)-\rho\|_1 & = \left\|\frac{1}{t} \int_0^t \mathrm{d}u \, \mathrm{e}^{-iHu} \bigl(\rho(0) - \rho\bigr) \mathrm{e}^{iHu}\right\|_1\\
    &\le \frac{1}{t} \int_0^t \mathrm{d}u \left\| \mathrm{e}^{-iHu} \bigl(\rho(0) - \rho\bigr) \mathrm{e}^{iHu}\right\|_1\\
    &= \frac{1}{t} \int_0^t \mathrm{d}u \|\rho(0) - \rho\|_1\\
    & \le \epsilon\,.
\end{split}
\end{align}

Next we show that the commutator of the time-averaged state with the Hamiltonian is decreasing with time. 
Using the equations of motion and the fundamental theorem of calculus, we have
\begin{align}
\begin{split}
    \|[H,\rhoav(t)] \|_1 &= \frac{1}{t}\|\rho(0) - \rho(t)\|_1\\
	&=\frac{1}{t}\|\left(\rho(0) - \rho\right) + \left(\rho - \rho(t)\right)\|_1\,.
\end{split}
\end{align}
Then again using the triangle inequality and the unitary invariance of the norm we have
\begin{align}
\begin{split}
	\|[H,\rhoav(t)] \|_1&\leq\frac{\epsilon}{t} + \frac{1}{t}\|\rho-\rho(t)\|_1\\
	&\leq\frac{\epsilon}{t} + \frac{1}{t}\|\rho-\rho(0)\|_1\\
	&\le \frac{2\epsilon}{t}
\end{split}
\end{align}
and the result is immediate.
\end{proof}

Although the time-averaged state becomes a better and better approximate state preparation, there is no physical way to prepare the exact state $\rhoav(t)$ for $t>0$.
An obvious approach to deal with this is to sample from a discrete set of intermediate times and average the values of the experiments at these sampled times.
Let us denote by $\rhoav_s(t)$ a discretely averaged density operator over a set of $s$ points $\{t_1\ldots,t_s\}$ in the interval $[0,t]$.
We will allow positive weights $w_i$ so that our approximate time-averaged state is given by
\begin{align}
\label{eq:rhoav-s}
    \rhoav(t) \approx \rhoav_s(t) := \sum_{i=1}^s w_i \rho(t_i)\,,
\end{align}
where $t_i = u_i t$ for some $w_i>0$ and $u_i \in (0,1)$ to be chosen later by \cref{eq:u,eq:weights}. 
The matrix elements of $K\bigl(\rhoav(t)\bigr)$ are given by
\begin{align}
\label{eq:K-rhoav}
    K\bigl(\rhoav(t)\bigr)_{jk} = \Tr\bigl(i[P_j,P_k]\rhoav(t)\bigr)\,,
\end{align}
which is linear in $\rhoav(t)$. 
This linearity is important because it means that a weighted average the results of experiments with the states $\rho(t_i)$ will have the same expected value as an experiment with the unphysical state $\rhoav(t)$. 

Our next result bounds the error in the 1-norm associated to approximating $\rhoav(t)$ from \cref{eq:timeavgrho} by $\rhoav_s(t)$ from \cref{eq:rhoav-s}. 
\begin{theorem}
\label{thm:quadrature}
There exist sets of weights $w_i>0$ and times $t_i\in [0,t]$, $\bigl\{(w_i,t_i)\bigr\}_{i=1}^s$, such that 
\begin{align}
\label{eq:quaderrorbound1norm}
    \bigl\|\rhoav(t) - \rhoav_s(t)\bigr\|_1 
    \le \frac{\sqrt{\pi}}{4\sqrt{s}} \biggl(\frac{\mathrm{e}\|H\|t}{4s}\biggr)^{2s}\,.
\end{align}
\end{theorem}

\begin{proof}
We first change variables so that
\begin{align}
    \rhoav(t) = \frac{1}{t}\int_0^t \rho(u) \mathrm{d}u = \int_0^1 \rho\bigl(ut\bigr)\mathrm{d}u\,.
\end{align}
Now trace both sides against any operator $\Pi$ satisfying $\|\Pi\|\le 1$, and we can introduce a function
\begin{align}
\label{eq:f(u)}
    f(u) = \Tr\bigl(\Pi\, \rho(ut)\bigr)\,.
\end{align}
The function $f$ is very well behaved, and is an entire function in the complex plane since is obtained from compositions, sums and products of entire functions. 
In particular, it has well-defined derivatives at all orders.

We will use Gauss-Legendre quadrature to construct positive weights $w_i$ and points $u_i$ so that 
\begin{align}
    I := \int_0^1 f(u) \mathrm{d}u\, \approx\, I_s:=\sum_{i=1}^s w_i f(u_i)\,.
\end{align}
We will choose the evaluation points $u_i$ from the distinct roots of the $s$th order Legendre polynomial $\mathcal{P}_s(x)$, so that for $i=1,\ldots,s$ we have
\begin{align}
\label{eq:u}
    \mathcal{P}_s(2u_i-1) = 0\,.
\end{align}
The choice of weight for point $i$ is given by
\begin{align}
\label{eq:weights}
    w_i = \frac{4 u_i(1-u_i)}{(s+1)^2 \mathcal{P}_{s+1}(2u_i-1)^2}\,.
\end{align}
With these choices, an upper bound for the error of the quadrature rule $I_s$ is given by~\cite[\S5.2]{Kahaner1989}
\begin{align}
    \bigl|I-I_s\bigr| \le \frac{(s!)^4}{(2s+1)((2s)!)^3} |f^{(2s)}(\xi)|\,,
\end{align}
for some $\xi \in [0,1]$. 
By \cref{lemma:M-bound} below and $\|\Pi\|\le 1$, the term $|f^{(2s)}(\xi)|$ is bounded by 
\begin{align}
    |f^{(2s)}(\xi)| \le \bigl(2\|H\|t\bigr)^{2s}\,
\end{align}
and we find the following inequality,
\begin{align}
    |I-I_s| \le \frac{\bigl(2\|H\|t\bigr)^{2s}(s!)^4}{(2s+1)((2s)!)^3}\,.
\end{align}
The right hand side of \cref{eq:quaderrorbound1norm} is an upper bound on this after an elementary application of Stirling's formula.

This bound holds for any choice of $\Pi$, so in particular
\begin{align}
    |I-I_s| \le \max_{\|\Pi\| \le 1} |I-I_s| = \bigl\|\rhoav(t) - \rhoav_s(t)\bigr\|_1
\end{align}
and the result follows.
\end{proof}

We remark that, while we have used Gauss-Legendre quadrature to get a provable guarantee on the approximation error for the time-averaged matrix elements, it would be more suitable in practice to use a Gauss-Kronrod quadrature formula or other nested quadrature rule so that convergence can be checked online while reusing preexisting data points. 

We now prove the lemma used in the proof of \cref{thm:quadrature}.

\begin{lemma}
\label{lemma:M-bound}
For any operator $\Pi$, we have the bound
\begin{align}
    \sup_{u \in [0,1]}\, \biggl|\frac{\partial^k}{\partial u^k} \Tr\bigl(P\rho\bigl(ut\bigr)\bigr)\biggr|\le \|\Pi\| \bigl(2\|H\|t\bigr)^k\,.
\end{align}
\end{lemma}

\begin{proof}
From the equation of motion and the chain rule, we have
\begin{align}
    \partial_u \rho\bigl(ut\bigr) = -i\bigl[H,\rho\bigl(ut\bigr)\bigr] t\,.
\end{align}
Therefore the $k$th derivative involves a $k$-nested commutator which will have $2^k$ terms, and there will be an overall factor of $t^k$. 
Expanding the nested commutators and using the triangle inequality, each term has $k$ factors of $H$ in it, and is of the form $|\Tr(P H^{k-\ell} \rho H^{\ell})|$.
By using the matrix H\"{o}lder inequality and the submultiplicativity of the norm, each term is less than $\|\Pi\| \|H\|^k \|\rho\|_1 \le \|\Pi\|\|H\|^k$.
Summing all of these terms and accounting for the overall factor from the chain rule gives the result. 
\end{proof}

This additive error on the estimates to the matrix elements of $K$ can be incorporated into the bounds derived in \cref{appendix:sample-complexity} for finite sampling. 

To see how approximate state preparations will affect our estimate, we define a measure of overlap between two subspaces spanned by vectors $a$ and $b$ given by the \emph{fidelity}, 
\begin{align}
    F(a,b) = \frac{ |a\T b|^2}{\|a\|_2^2 \|b\|_2^2}\,.
\end{align}
Note that this measure is canonical in the sense that it only depends on the subspace projectors and not on the specific choice of spanning vector within each subspace. 
That is, it is invariant under the rescalings $a\to \alpha a$ and $b \to \beta b$ for any nonzero $\alpha, \beta$. 
The fidelity is always bounded by $F\in[0,1]$, with $F=1$ if and only if $a=b$. 
Therefore $1-F$, called the \emph{infidelity}, is a sensible measure of error.

\begin{theorem}
\label{thm:bound}
Let $\rho$ be a $\delta$-approximate steady state for $H$, satisfying $\bigl\|[H,\rho]\bigr\|_1\le\delta$. 
Then the estimate $\hat{c}$ obtained from the least right singular vector of $K(\rho)$ obeys an error bound with $c=c(H)$ given by
\begin{align}
    1-F(c,\hat{c}) \le \frac{m\delta^2}{\Delta\|c\|_2^2}\,,
\end{align}
where $m = \dim(c)$ and $\Delta > 0$ is the spectral gap of $K(\rho)$. 
\end{theorem}

\begin{proof}
Since $\rho$ is a $\delta$-approximate steady state with $H$, we have
\begin{align}
    \bigl\|[H,\rho]\bigr\|_1\le\delta\,.
\end{align}
The matrix elements of $K = K(\rho)$ are given by
\begin{align}
    K_{jk} = i \Tr\bigl([P_j,P_k]\rho\bigr)\,,
\end{align}
where the $P_j$ are the Pauli matrices and the indices run over the supported elements in the span that we are considering. 
The unknown $H$ is described by $c(H)$, the vector of couplings of $H$ with elements
\begin{align}
    c(H)_j = \frac{1}{d}\Tr(P_j H)\,.
\end{align}
so that
\begin{align}
    H = \sum_j c(H)_j P_j\,.
\end{align}

We first show that $c = c(H)$ is an approximate null vector of $K$.
Looking at the $j$th element of $Kc$, we find from the cyclic property of the trace that
\begin{align}
\begin{split}
    \bigl[Kc\bigr]_j & = \sum_{k} i\Tr\bigl([P_j,P_k]\rho\bigr) c_k\\
    & = i\Tr\bigl([P_j,H]\rho\bigr)\\
    & = i\Tr\bigl(P_j H \rho - H P_j \rho\bigr)\\
    & = i\Tr\bigl(P_j [H,\rho]\bigr)\,.
\end{split}
\end{align}
Now using the matrix H\"{o}lder inequality, we have 
\begin{align}
\label{eq:inftynormbound}
\begin{split}
    \bigl|\bigl[Kc\bigr]_j\bigr| & \le \|P_j\|_\infty\, \bigl\|[H,\rho]\bigr\|_1\\
    & \le \delta\,.
\end{split}
\end{align}

Now consider the estimate $\hat{c}$ that would be returned by finding the least right singular vector of $K$. 
We let $M=K\T K$ and then $\hat{c}$ is equivalently the eigenvector of $M$ with the least eigenvalue. 
We have assumed that the gap of $K$ is $\Delta > 0$, so the ``ground state'' of $M$ is unique and spanned by $\hat{c}$. 

Now from \cref{eq:inftynormbound} we have that the correct unknown vector $c$ has small overlap with $M$,
\begin{align}
\label{eq:cMc-ineq}
    c\T M c \le m\delta^2\,.
\end{align}
Next, we note that as $M$ is a positive semidefinite matrix, we have the matrix inequality
\begin{align}
\label{eq:M-ineq}
    M \succeq \Delta \left(\mathbbm{1}-\frac{\hat{c}\,\hat{c}\T}{\|\hat{c}\|_2^2}\right)\,.
\end{align}
Now multiplying \cref{eq:M-ineq} by $c\,c\T$ and taking the trace and then using \cref{eq:cMc-ineq}, we find a bound on the infidelity of
\begin{align}
    1-F(c,\hat{c}) \le \frac{m\delta^2}{\Delta\|c\|_2^2}\,,
\end{align}
as claimed. 
\end{proof}

This error bound agrees with the scaling estimate derived in Ref.~\cite{qi2019determininglocal} using a matrix perturbation theory argument.
Our derivation has the advantage that it makes clear the scaling with respect to the length of the true Hamiltonian vector $c(H)$, and our bound is an effective bound that contains no uncontrolled sub-leading error terms. 

As a corollary of this result, suppose we have estimates of $K(\rho)$ that additionally have some additive error, $K \to K+\mathcal{E}$, where $|\mathcal{E}_{jk}| \le \epsilon$. 
Then the same argument as before shows that the true Hamiltonian $c(H)$ is close to the approximate kernel estimate $\hat{c}$.

\begin{corollary}
\label{cor:KE}
Under the same conditions as \cref{thm:bound}, if $K\to K+\mathcal{E}$ and $|\mathcal{E}_{jk}| \le \epsilon$ then 
\begin{align}
    1-F(c,\hat{c}) \le \frac{m(\delta+\epsilon\|c\|_1)^2}{\Delta\|c\|_2^2}\,.
\end{align}
where $\Delta$ is the gap of $K+\mathcal{E}$.
\end{corollary}

\begin{proof}
The idea is the same as the proof of \cref{thm:bound}, except we transform $M\to M+\mathcal{E}\T K+K\T\mathcal{E}+\mathcal{E}\T\mathcal{E}$ on the left hand side of \cref{eq:cMc-ineq}. 
Then using the triangle inequality and the elementary estimates
\begin{align}
    |c\T\mathcal{E}\T Kc| \le m\delta\epsilon\|c\|_1\ \, \text{and}\ \,  |c\T\mathcal{E}\T\mathcal{E}c|\le m\epsilon^2\|c\|_1^2\,,
\end{align}
the result follows.
\end{proof}

\section{Pseudocode for the main algorithms}\label{appendix:algorithms}

The first algorithm (\cref{alg:efficient-pauli-expectation}) presented in this appendix allows for the fast estimation of $k-$body  Pauli expectation values. 
Given a list of expectation values required $C_{1,\dots,m}$ and an input state $\rho$, it will return and estimate of $\Tr\bigl(C_i \rho\bigr)$ for each of these values.

The second algorithm (\cref{alg:tim-averaged-states}) leverages the first to show how to measure the relevant expectation values for a time-averaged state.

Both of these algorithms deliver the performance guarantees set out in \cref{appendix:approximate-state-preps} and \cref{appendix:sample-complexity}.
The pseudocode is shown in \cref{tab:pseudocode}. 

\begin{table*}[t!]
~\hfill
\begin{minipage}[t]{.97\columnwidth}
    \begin{algorithm}[H]
	\caption{\!\textbf{:}\label{alg:efficient-pauli-expectation} $\textsc{Pauli}(\rho,C,L)$\\
	Efficient estimation of Pauli expectation values}
	\begin{algorithmic}[1]
		\Require $\rho$, $C$, $L$.\\
		\# $\rho$ is an $n$-qubit quantum state\\
		\# $C$ is a list of $m$ Pauli operators $\{C_1,\dots,C_{m}\}$
		\State $M\gets $ Initialise vector with $L$ entries.
		\State $P\gets $ Initialise vector with $L$ entries.
		\For{$l=1,\dots,L$}
		\State $P_l\gets$ Random full-weight Pauli
		\State $M_l\gets$ length $n$ array of $\pm 1$ meas.\ results for $\rho$, $P_l$
		\EndFor
		\State E $\gets $ vector with $m$ entries, each initialised to 0.
		\For{$k=1,\dots,m$}
		\ForAll{$l$ such that $P_l$ supports $C_k$}
		\State $E_k\gets E_k + \prod_{j \in \text{supp}(C_k)} M_{l,j}$
		\EndFor
		\State $j\gets $ number of entries in $P$ that support $C_k$.
		\State $E_k\gets E_k/j$
		\EndFor
		\Ensure $E$ \Comment Approximates $\Tr\bigl(C_i \rho\bigr)$ for all $i$.
	\end{algorithmic}
\end{algorithm}
\end{minipage}
\quad\hfill\quad
\begin{minipage}[t]{.97\columnwidth}
\begin{algorithm}[H]
	\caption{\!\textbf{:}\label{alg:tim-averaged-states} $\textsc{TimeAvg}(\rho(0),C,s,t,L)$\\
	Approximate time-averaged Pauli expectation values}
	\begin{algorithmic}[1]
		\Require $\rho(0)$, $C$, $s$, $t$, $L$\\
        \# $C$ is a list of $m$ Pauli operators $\{C_1,\dots,C_{m}\}$\\
        \# Compute quadrature weights and times \\
		$u \gets \bigl\{u_k: \mathcal{P}_s(2u_k-1)=0\,|\,k=1:s\bigr\}$  \Comment \cref{eq:u}\\
		$w \gets \bigl\{w_k = \tfrac{4u_k(1-u_k)}{(s+1)^2 \mathcal{P}_{s+1}(2u_k-1)^2}\,|\,k=1:s\bigr\}$  \Comment \cref{eq:weights}\\
		
		$\tau \gets ut$\\
		$M \gets$ Initialise vector with $s$ entries.
		\For{$k=1,\dots,s$}
		\State $\rho \gets \rho(0)$ \Comment Initialize.
		\State $\rho \gets \rho(\tau_k) =\mathrm{e}^{-iH\tau_k}\rho\mathrm{e}^{iH\tau_k}$  \Comment Time evolve.
		\State $M_k\gets$ \textsc{Pauli}$\left(\rho,C,L\right)$ \Comment \Cref{alg:efficient-pauli-expectation}.\\
		\# $M_k$ contains estimates of $\Tr\bigl(C_i\rho(\tau_k)\bigr)$ for all $i$.
		\EndFor
		\State $E\gets$ Initialise empty vector with $m$ entries.
		\For{$i=1,\dots,m$}
		    \State $E_i = \sum_{k=1}^{s} w_k M_{k,i}$ \Comment \cref{eq:rhoav-s}
		\EndFor
		\Ensure $E$ \Comment Approximates $\Tr\bigl(C_i \rhoav(t)\bigr)$ for all $i$.
	\end{algorithmic}
\end{algorithm}
\end{minipage}
\hfill~
\caption{Pseudocode for the efficient estimation of Pauli expectation values and for approximate time-averaging of Pauli expectation values. 
The algorithms are analyzed in \cref{appendix:approximate-state-preps} and \cref{appendix:sample-complexity}.
The error in the estimates is controlled by the input parameter $L$, and to fix a constant error, $L$ should grow as $3^k$ where $k$ is the weight of the highest-weight Pauli in $C$.}
\label{tab:pseudocode}
\end{table*}

\section{Sample complexity}\label{appendix:sample-complexity}

If the Hamiltonian $H$ is supported on a basis of Pauli operators $\{P_j\}_{j=1}^m$ where each of the $m$ operators has weight at most $k$, then with probability at least $1-\delta$ we can measure all of the relevant elements of the matrix $K$ to precision $\epsilon$ using a small number of state preparations. 
Using \cref{thm:random} below, we require at most $\tfrac{2}{\epsilon^2(1-\epsilon)} 3^{2k-1} \log\bigl(\tfrac{3m'}{\delta}\bigr)$ state preparations.
Here $m' \le m(m-1)/2$ is the number of nonvanishing commutators in the set $\bigl\{[P_j,P_k]\bigr\}_{j,k=1}^m$, and the weight $\mathrm{wt}\bigl([P_j,P_k]\bigr)\le 2k-1$ for $k$-body Hamiltonians.
A slightly weaker result using $\epsilon^{-2}\mathrm{e}^{O(k)} \log^2 m$ total measurements can be obtained by using the recent work Ref.~\cite{cotler2019}. 
Here we give a simple randomized algorithm that avoids using perfect hash families and has an improved scaling with $m$. 

\begin{theorem}
\label{thm:random}
Consider a set of $m$ Pauli operators on $n$ qubits $\{P_i\}_{i=1}^m$, each with weight $\le k$ and let $\epsilon \in (0,1)$. 
Then with probability at least $1-\delta$, using 
\begin{align}
    N = \frac{2}{\epsilon^2(1-\epsilon)}3^{k}\log(3m/\delta)
\end{align}
copies of $\rho$ suffices to estimate $\Tr(P_i \rho)$ to within $\pm\epsilon$ for all $i$.
\end{theorem}

\begin{proof}
For each of $N$ state preparations we measure an independent, random, full-weight string of Pauli operators. 
For example, for three qubits, we might measure $XYZ$, then $YZY$, etc, with an independent choice of Pauli on each qubit for each copy.
A measurement of any given weight-$k$ Pauli operator will occur as a marginal measurement in a $p=\frac{1}{3^k}$ fraction of strings, at least in expectation. 
Let $S$ be the number of times that a given correlator appears in the list of $N$ strings. 
Then by the Chernoff bound, the probability of $S < T = (1-\epsilon)Np$ for some $\epsilon > 0$ is bounded from above by 
\begin{align}
    \Pr\bigl(S < T\bigr) \le \exp\bigl[-\epsilon^2 N p/2\bigr] < \exp\bigl[-\epsilon^2 T/2\bigr]\,.
\end{align}

Assuming that each of the $m$ correlators of interest appears in at least $T$ strings, we can average the results of those $T$ (or more) measurements to get estimates of the expectation values. 
Again by the Chernoff bound, the probability that the sample mean $\hat{P}$ averaged over $T$ independent trials is within $\epsilon$ of the true expected value $\langle P\rangle$ is bounded by
\begin{align}
    \Pr\bigl(|\hat{P}-\langle P\rangle| \ge \epsilon\bigr) \le 2 \exp\bigl[-\epsilon^2 T/2\bigr]\,.
\end{align}
Then by the union bound, the probability that any of the estimates $\hat{P}_i$ is further than $\epsilon$ from its mean is bounded by
\begin{align}
    \Pr\Bigl(\max_i \bigl|\hat{P}_i-\langle P_i\rangle\bigr| \ge \epsilon\Bigr) < 3 m \exp\bigl[-\epsilon^2 T/2\bigr]\,.
\end{align}
Therefore, to ensure that the total probability of failure is less than $\delta$, it suffices to choose $T = 2\epsilon^{-2} \log(3m/\delta)$.
\end{proof}

We can now bound the entire sample complexity of the protocol to get a high-fidelity kernel estimate. 
It is an open question how to extend this to a bound on the Bayesian MAP estimator that we develop here, but as our numerics suggest, the sample complexity of the MAP estimator should generally be better. 
Our main result is the following.

\begin{theorem}
\label{thm:main}
Using \cref{alg:tim-averaged-states} to construct the approximate kernel estimate $\hat{c}$, and given an \textit{a priori} upper bound $|c_i|=O(1)$, then with probability at least $1-\delta$ using
\begin{align}
    N = O\left(\frac{m^3 3^{2k}}{\varepsilon^{3/2} \Delta^{3/2}} \sqrt{\log\Bigl(\tfrac{m^3 }{\delta\sqrt{\varepsilon\Delta} }\Bigr)}\right)
\end{align}
samples is sufficient to get an infidelity $1-F(c,\hat{c}) \le \varepsilon$ of a $k$-body Hamiltonian $H$ with support on a set of $m$ given Pauli operators, where $m = \dim(c)$ and $\Delta$ is the spectral gap of the noisy $K$ matrix.
\end{theorem}

\begin{proof}
We assume an initial state preparation that is close to a steady state as follows:
\begin{align}
    \|\rho(0)-\rho\|_1 \le \eta\,.
\end{align}
By \cref{lemma:preparation-error}, the time-averaged state $\rhoav(t)$ is then a $(2\eta/t)$-approximate steady state. 
We can approximate the matrix elements of $K$ on this state by repeated sampling via \cref{thm:random} and by using the quadrature formulas from \cref{thm:quadrature}. 
For each of the sample points $t_i$, $i=1,\ldots,s$ we can estimate each of the $m'\le m(m-1)/2$ nontrivial matrix elements of $K$ using
\begin{align}
\label{eq:Ntotal}
    N = \frac{2 s}{\epsilon^2(1-\epsilon)}3^{2k-1} \log\biggl(\frac{3m's}{\delta}\biggr) 
\end{align}
total copies and the results are guaranteed with probability $1-\delta$ to be within $\pm\epsilon$.
Now let 
\begin{align}
    \mu := \frac{\sqrt{\pi}}{4\sqrt{s}} \biggl(\frac{\mathrm{e}ht}{4s}\biggr)^{2s}\,,
\end{align}
where $h\ge \|c\|_1 \ge \|H\|$ is any \textit{a priori} upper bound on the 1-norm of the unknown vector $c$ and, by the triangle inequality, the norm of $H$. 
Since the weights $w_i$ in \cref{eq:weights} satisfy $\sum_i w_i = 1$, then the 
error $\mathcal{E}_{jk}$ of each matrix element of $K$ is bounded by
\begin{align}
    |\mathcal{E}_{jk}|\le \mu + \epsilon\,,
\end{align}
where the $\mu$ contribution comes from quadrature error and the $\epsilon$ comes from sampling. 

Under these conditions, by \cref{cor:KE} the approximate kernel estimator returns an estimate $\hat{c}$ that has infidelity bounded by
\begin{align}
    1-F(c,\hat{c}) \le \frac{m h^2}{\Delta\|c\|_2^2}\biggl(\frac{2\eta}{h t}+\mu+\epsilon\biggr)^2\,.
\end{align}
where $\Delta$ is the gap of $K+\mathcal{E}$.
This expression hides some of the dependence on $t$ and $s$, but we can optimize the choice of $t$ to minimize the error for fixed values of $h, \eta, s$, and $\epsilon$.  

The optimal time $t^\star$ while holding the other parameters constant can be calculated as
\begin{align}
    h t^\star = 4s \biggl(\frac{s^{-\frac{3}{2}}\eta}{\mathrm{e}^{2s}\sqrt{\pi}}\biggr)^{\frac{1}{1+2s}} = \Theta(s)\,.\label{eq:optimalT}
\end{align}
Plugging in this value, we find
\begin{align}
    \frac{2 \eta }{h t}+\mu\bigg\rvert_{t=t^\star} &= (\sqrt{\pi } e^{2 s} s^{3/2} \eta ^{2 s})^{\frac{1}{1+2s}} \frac{2 s+1}{4 s^2}\nonumber\\
    & < \frac{7}{2s}\,,
\end{align}
where the latter bound uses $\eta \le 2$. 
Now using $\|c\|_1 \le \sqrt{m}\|c\|_2$, we can substitute this upper bound for $h$ in the scaling to cancel the dependence on $\|c\|_2$, and we have
\begin{align}
\label{eq:F-scaling-second}
     1-F(c,\hat{c}) &\le \frac{m^2}{\Delta}\biggl(\frac{7}{2s}+\epsilon\biggr)^2\,.
\end{align}

This uses at most
\begin{align}
    L = \frac{2}{\epsilon^2 (1-\epsilon)} 3^{2k-1} \log\left(\frac{3m's}{\delta}\right)
\end{align}
many samples for each of the $s$ quadrature points. 
A slightly weaker bound than claimed comes from choosing $\epsilon = 7/2s$ and $s = \left\lceil\frac{7 m}{\sqrt{\Delta}\sqrt{\varepsilon}}\right\rceil$. Then we find that $1-F(c,\hat{c}) \le \varepsilon$ using at most
\begin{align}
    L = O\left(\frac{m^2 3^{2k}}{\varepsilon \Delta} \log\left(\frac{m^3}{\delta\sqrt{\varepsilon\Delta}}\right)\right)
\end{align}
measurements per quadrature point, where we use $m' < m^2/2$. 
The total number of measurements is then
\begin{align}
    N = O\left(\frac{m^3 3^{2k}}{\varepsilon^{3/2} \Delta^{3/2}} \log\left(\frac{m^3}{\delta\sqrt{\varepsilon\Delta}}\right)\right)\,.
\end{align}

The square-root improvement in the logarithmic term is achieved by additionally optimizing the tradeoff between $\epsilon$ and $s$. 
Let $\alpha = O(m^2\delta^{-1})$ and $L' = 3^{-k}\sqrt{L}$, then we can choose $s =O\bigl(L'\sqrt{\log(L'\alpha)}\bigr)$ and $\epsilon = \sqrt{\log(\alpha s)}/L'$. 
Then choosing a number of samples per quadrature point of $L=3^{2k}m^2 \varepsilon^{-1}\Delta^{-1}\log(\alpha^2\varepsilon^{-1}\Delta^{-1})$ gives the stated result. 
\end{proof}

Finally, we state a corollary of \cref{thm:main} for the case of $k$-body or $k$-local Hamiltonians. 

\begin{corollary}
Under the same conditions as \cref{thm:main}, the sample complexity for \cref{alg:tim-averaged-states} to obtain $1-F(c,\hat{c}) \le \varepsilon$ is 
\begin{align}
    N_{\text{$k$-body}} = O\biggl(\frac{n^{3k}}{\varepsilon^{3/2} \Delta^{3/2}} \sqrt{\log\Bigl(\tfrac{n^{3k}}{\delta\sqrt{\varepsilon\Delta} }\Bigr)}\biggr)
\end{align}
for general $k$-body Hamiltonians and 
\begin{align}
    N_{\text{$k$-local},D} = O\biggl(\frac{n^3 k^{3D}3^{2k}}{\varepsilon^{3/2} \Delta^{3/2}} \sqrt{\log\Bigl(\tfrac{n^3 k^{3D}}{\delta\sqrt{\varepsilon\Delta} }\Bigr)}\biggr)
\end{align}
for $k$-local Hamiltonians in $D$-spatial dimensions. 
\end{corollary}
\begin{proof}
We simply observe that a general $k$-body Hamiltonian is supported on at most $m \le O(n^k)$ terms, and a $k$-local Hamiltonian in $D$ spatial dimensions has $m = O(n k^D)$ terms. 
\end{proof}

\bibliographystyle{apsrev4-1}
\bibliography{HamLearn}

\end{document}